\pdfoutput=1
\ifx\synctex\undefined\else\synctex=1\fi
\documentclass{llncs}

\usepackage[utf8]{inputenc}
\usepackage[T1]{fontenc}

\usepackage{microtype}

\usepackage{txfonts}

\usepackage{amssymb}
\usepackage{mathtools} %

\usepackage{array}

\usepackage{tikz}
\usetikzlibrary{positioning}
\usetikzlibrary{arrows,automata}
\usetikzlibrary{calc}
\usetikzlibrary{decorations.pathreplacing}
\usetikzlibrary{matrix}
\usetikzlibrary{arrows}

\usepackage[draft]{commenting}
\usepackage{paralist}

\usepackage[pdftex,%
 	colorlinks,%
 	hidelinks,
 	pdfauthor={Antoine Durand-Gasselin, Javier Esparza, Pierre Ganty, Rupak Majumdar},%
 	hypertexnames=false,%
	pdftitle={Model Checking Parameterized Asynchronous Shared-Memory Systems}]{hyperref}

\makeatletter
\def\verbatim{\small\@verbatim \frenchspacing\@vobeyspaces \@xverbatim}
\makeatother

{\itshape}{}

\newcommand{\Act}{{\cal A}}
\newcommand{\D}{{\cal D}}
\newcommand{\C}{{\cal C}}
\renewcommand{\S}{{\cal S}}
\newcommand{\G}{{\cal G}}
\newcommand{\N}{{\cal N}}
\newcommand{\betweenp}[1]{{\between}_{#1} }

\newcommand{\rb}{\ensuremath{r_{\star}}}
\newcommand{\wb}{\ensuremath{w_{\star}}}

\def\prod{\mathcal{P}}

\def\set#1{{\left\{ #1 \right\}}}

\newcommand{\proj}{\mathit{Proj}}

\newcommand{\by}[1]{\xrightarrow{#1}}
\newcommand{\absby}[1]{\xrightarrow{#1}_\alpha}
\newcommand{\ds}[1]{{\it ds}(#1)}
\newcommand{\ap}[1]{{\it ap}(#1)}

\newcommand{\dom}{\textrm{dom}}

\renewcommand{\acute}[1]{\overrightarrow{#1}}
\renewcommand{\grave}[1]{\overleftarrow{#1}}

\declareauthor{rm}{Rupak}{blue}
\declareauthor{pg}{Pierre}{red}
\declareauthor{je}{Javier}{black!30!green}
\declareauthor{adg}{Antoine}{black!30!yellow}

\title{\texorpdfstring{Model Checking Parameterized Asynchronous Shared-Memory Systems}{Model Checking Parameterized Asynchronous Shared-Memory Systems}}
\author{Antoine Durand-Gasselin\inst{1} \and Javier Esparza\inst{1} \and Pierre Ganty\inst{2} \and Rupak Majumdar\inst{3}}
\institute{$^1$TU Munich \quad $^2$IMDEA Software Institute \quad $^3$MPI-SWS}

\begin{document}

\maketitle

\begin{abstract}
\makeatletter{}%
%
%
%
%

We characterize the complexity of liveness verification for
parameterized systems consisting of a leader process and arbitrarily many
anonymous and identical contributor processes.  Processes communicate through a shared,
bounded-value register. While each operation on the register is atomic, there
is no synchronization primitive to execute a sequence of operations atomically.

We analyze the case in which processes are
modeled by finite-state machines or pushdown machines and the property is given
by a B\"uchi automaton over the alphabet of read and write actions of the leader.
We show that the problem is decidable, and has a surprisingly low complexity: 
it is NP-complete when all processes are finite-state machines, and is PSPACE-hard and in NEXPTIME 
when they are pushdown machines. This complexity is lower than for the 
non-parameterized case: liveness verification of finitely many finite-state machines is 
PSPACE-complete, and undecidable for two pushdown machines. 

For finite-state machines, our proofs characterize infinite behaviors using existential abstraction
and semilinear constraints.
For pushdown machines, we show how contributor computations of high stack height can be simulated by computations of
many contributors, each with low stack height. 
Together, our results characterize the complexity of verification 
for parameterized systems under the assumptions of anonymity and asynchrony.
 %

\end{abstract}

\makeatletter{}%
%
%
%

\section{Introduction}

We study the verification problem for
\emph{parameterized asynchronous shared-memory systems} \cite{Hague11,egm13}. 
These systems consist of a \emph{leader} process and
arbitrarily many identical \emph{contributors}, processes with no identity, running at
arbitrary relative speeds.%
The shared-memory consists of a read/write register that all processes can access
to perform either a read operation or a write operation. The register is
bounded: the set of values that can be stored is finite. Read/write operations execute 
atomically but sequences of operations do not: no process can conduct an atomic sequence 
of reads and writes while excluding all other processes. 
In a previous paper \cite{egm13}, we have studied the complexity of safety verification, 
which asks to check if a safety property holds no matter how many contributors are present.
In a nutshell, we showed that the problem is coNP-complete when both leader and contributors
are finite-state automata and PSPACE-complete when they are pushdown automata.

In this paper we complete the study of this model by addressing the verification of liveness properties 
specified as $\omega$-regular languages (which in particular encompasses LTL model-checking). 
Given a property like ``every request is eventually granted''
and a system with a fixed number of processes, one is often able to guess an upper bound
on the maximal number of steps until the request is granted, and replace the property
by the safety property ``every request is granted after at most $K$ steps.'' 
In parameterized systems this bound can depend on the (unbounded) number of processes, and so 
reducing liveness to safety, or to finitary reasoning, is not obvious. 
Indeed, for many parameterized models, liveness verification is undecidable even if safety is decidable \cite{EFM99,rmeyer2008}.

Our results show that there is no large complexity gap between liveness and
safety verification: liveness verification (existence of an infinite computation
violating a property) is NP-complete in the finite-state
case, and PSPACE-hard and in NEXPTIME in the pushdown case.  
In contrast, remember that liveness checking is already PSPACE-complete for a {\em finite} number of finite-state
machines, and undecidable for a {\em fixed} number of pushdown systems.
Thus, not only is liveness verification decidable in the parameterized setting but 
the complexity of the parameterized problem is {\em lower} 
than in the non-parameterized case, where all processes are part of the
input. We interpret this as follows: in asynchronous shared-memory systems, the
existence of arbitrarily many processes leads to a ``noisy'' environment, in
which contributors may hinder progress by replying to past messages from the
leader, long after the computation has moved forward to a new phase. It is known
that imperfect communication  can {\em reduce} the power of computation and the complexity of verification
problems: the best known example are lossy channel systems, for which many verification
problems are decidable, while they are undecidable for perfect channels 
(see e.g. \cite{AbJo:lossy:IC,Parosh:etal:attractor:IC}). 
Our results reveal another instance of the same phenomenon.

Technically, our proof methods are very different from those used for safety
verification. 
Our previous results \cite{egm13} relied on a fundamental Simulation Lemma,
inspired by Hague's work \cite{Hague11}, stating that the {\em finite} behaviors of an arbitrary number of contributors can be simulated by a finite 
number of {\em simulators}, one for each possible value of the register. 
Unfortunately, the Simulation Lemma 
does not extend to infinite behaviors, and so we have to develop new ideas.
In the case in which both leader and contributors are finite-state machines, 
the NP-completeness
result is obtained by means of a combination of an abstraction
that overapproximates the set of possible infinite behaviors, and a semilinear constraint that
allows us to regain precision. The case in which both leader and contributors are 
pushdown machines is very involved. In a nutshell, we show that pushdown runs 
in which a parameter called the \emph{effective stack height} grows too much
can be ``distributed'' into a number of runs with smaller effective stack height.
We then prove that the behaviors of a pushdown machine with a bounded effective stack height 
can be simulated by an exponentially larger finite-state machine.

\noindent
\emph{Related Work.}
Parameterized verification has been studied extensively, both theoretically and practically.
While very simple variants of the problem are already undecidable \cite{AxKozen86}, 
many non-trivial parameterized models retain decidability.
There is no clear ``rule of thumb'' that allows one to predict what model checking 
problems are decidable, nor their complexities, other than ``liveness is generally harder than safety.''
For example, coverability for Petri nets---in which finite-state, identityless processes communicate via rendezvous or global shared state---
is EXPSPACE-complete, 
higher than the PSPACE-completeness of the non-parameterized version, and
verification of liveness properties can be equivalent to Petri net reachability, 
for which we only know non-primitive recursive upper bounds, or even undecidable.
Safety verification for extensions to Petri nets with reset or transfer, or broadcast protocols, where arbitrarily many finite-state processes communicate
through broadcast messages, 
are non-primitive recursive; liveness verification is undecidable in all cases \cite{ACJT96,EFM99,rmeyer2008}.
Thus, our results, which show simultaneously lower complexity than non-parameterized problems, as well as similar complexity for
liveness and safety, are quite unexpected.

German and Sistla \cite{GS92} and Aminof {\em et al.}~\cite{AminofKRSV14} 
have studied a parameterized model with rendezvous as communication primitive, where processes are 
finite-state machines. Model checking the fully symmetrical case---only contributors, no leaders---runs in 
polynomial time (other topologies have also been considered \cite{AminofKRSV14}),
while the asymmetric case with a leader is EXPSPACE-complete. 
In this paper we study the same problems, but for a shared memory communication primitive. 

\emph{Population protocols} \cite{Angluin} are another well-studied model of identityless 
asynchronous finite-state systems communicating via rendezvous. The semantics 
of population protocols is given over fair runs, in which every potential interaction that is infinitely often enabled is infinitely often taken. With this semantics, population protocols compute exactly the semilinear predicates \cite{Angluin}. In this paper we do not study what our model can compute (in particular, we are agnostic with respect to which fairness assumptions are reasonable),
but what we can compute or decide about the model.

%
%
%
%

 %

%
\makeatletter{}%
%
%
%

\section{Formal Model: Non-Atomic Networks}
\label{sec:prelim}

In this paper, we identify systems with languages. System actions are
modeled as symbols in an alphabet, executions are modeled as infinite words,
and the system itself is modeled as the
language of its executions. Composition operations that combine systems into larger ones
are modeled as operations on languages.

\subsection{Systems as languages}%
An \emph{alphabet} \(\Sigma\) is a finite, non-empty set of \emph{symbols}. A
\emph{word} over \(\Sigma\) is a finite sequence over \(\Sigma\) including the
empty sequence denoted \(\varepsilon\), and a \emph{language} is a set of
words. An \emph{$\omega$-word} over \(\Sigma\) is an infinite sequence of
symbols of \(\Sigma\), and an \emph{$\omega$-language} is a set of
$\omega$-words.  We use \(\Sigma^*\) (resp.\ \(\Sigma^\omega\)) to denote the
language of all words (resp.\ $\omega$-words) over \(\Sigma\).  When there is
no ambiguity, we use ``words'' to refer to words or \(\omega\)-words.  We do
similarly for languages.  Let \(w\) be a sequence over some alphabet, define
\(\dom(w)=\{1,\ldots,n\}\) if \(w = a_1 a_2 \ldots a_n\) is a word; else (\(w\)
is an \(\omega\)-word) \(\dom(w)\) denote the set \(\mathbb{N}\setminus\{0\}\).
Elements of \(\dom(w)\) are called \emph{positions}.  The \emph{length} of a
sequence \(w\) is defined to be \(\sup \dom(w)\) and is denoted
\(\vert{w}\vert\).  We denote by \((w)_i\) the symbol of \(w\) at position
\(i\) if \(i\in\dom(w)\), \(\varepsilon\) otherwise. Moreover, let
\((w)_{i..j}\) with \(i,j\in\mathbb{N}\) and \(i<j\) denote \( (w)_{i}
(w)_{i+1} \ldots (w)_{j}\). Also \( (w)_{i..\infty}\) denotes \(
(w)_{i}(w)_{i+1}\ldots\) For words $u, v\in(\Sigma^\omega\cup\Sigma^{*})$, we
say $u$ is a \emph{prefix} of $v$ if either $u = v$ or $u \in \Sigma^*$ and
there is a $w\in(\Sigma^\omega\cup\Sigma^{*})$ such that $v = u w$.

\paragraph{Combining systems: Shuffle.} %
Intuitively, the shuffle of systems \(L_1\) and \(L_2\) is the system
interleaving the executions of \(L_1\) with those of \(L_2\). 
Given two $\omega$-languages \(L_1 \subseteq \Sigma_1^\omega\) and 
\(L_2 \subseteq \Sigma_2^\omega\), their \emph{shuffle}, denoted by \(L_1 \between L_2\), 
is the \(\omega\)-language over \((\Sigma_1 \cup \Sigma_2)\) defined as 
follows.
Given two $\omega$-words \(x \in \Sigma_1^\omega, y \in \Sigma_2^\omega\), we say that \(z \in (\Sigma_1 \cup \Sigma_2)^\omega \) 
is an \emph{interleaving} of \(x\) and \(y\) if there exist (possibly empty) words 
\(x_1, x_2, \ldots, x_i, \ldots \in \Sigma_1^*\) and \( y_1, y_2, \ldots, y_i, \ldots \in \Sigma_2^*\) such that 
each \(x_1 x_2 \cdots x_i\) is a prefix of \(x\),
and each \(y_1 y_2 \cdots y_i\) is a prefix of \(y\),
and \(z=x_1y_1 x_2y_2 \cdots x_iy_i \cdots \in \Sigma^\omega\) is an \(\omega\)-word.
Then \(L_1 \between L_2 = \textstyle{\bigcup_{x\in L_1, y\in L_2}} x \between
y\), where \(x \between y \) denotes the set of all interleavings of \(x\) and
\(y\).  For  example, if \(L_1 = ab^\omega\) and \(L_2 = ab^\omega\), we get \(L_1 \between L_2 = (a+ab^*a)b^\omega\). 
Shuffle is associative and commutative, and so we can 
write \(L_1 \between \cdots \between L_n\) or \({}\between_{i=1}^n L_i\).

\paragraph{Combining systems: Asynchronous product.} %
The asynchronous product of \(L_1 \subseteq \Sigma_1^\omega\) and \(L_2
\subseteq \Sigma_2^\omega\) also interleaves the executions but, this time,
the actions in the common alphabet must now be executed jointly.
The \(\omega\)-language of the resulting system, called the \emph{asynchronous
product} of \(L_1\) and \(L_2\), is denoted by \(L_1 \parallel L_2\), and
defined as follows.  Let \(\proj_{\Sigma}(w)\) be the word obtained by erasing
from \(w\) all occurrences of symbols not in \(\Sigma\). \(L_1 \parallel
L_2\) is the \(\omega\)-language over the alphabet
\(\Sigma=\Sigma_1\cup\Sigma_2\) such that \(w\in L_1 \parallel L_2\) if{}f
\(\proj_{\Sigma_1}(w)\) and \(\proj_{\Sigma_2}(w)\) are prefixes of words in
\(L_1\) and \(L_2\), respectively. 
We abuse notation and write \(w_1 \parallel L_2\)
instead of \(\{w_1\} \parallel L_2\) when \(L_1=\set{w_1}\).  For example, let
\(\Sigma_1 = \{a,c\}\) and \(\Sigma_2 = \{b,c\}\).  For \(L_1 = (ac)^\omega\)
and \(L_2 = (bc)^\omega\) we get \(L_1 \parallel L_2 = ((ab+ba)c)^\omega\). 
Observe that the language \(L_1\parallel L_2\) depends on \(L_1\), \(L_2\) and
also on \(\Sigma_1\) and \(\Sigma_2\). For example, if \(\Sigma_1 = \{a\}\) and
\(\Sigma_2 = \{b\}\), then \(\{a^\omega\} \parallel \{b^\omega\} =
(a+b)^\omega\), but if \(\Sigma_1 = \{a,b\}=\Sigma_2\), then \(\{a^\omega\}
\parallel \{b^\omega\} = \emptyset\).  So we should more properly write
\(L_1\parallel_{\Sigma_1,\Sigma_2} L_2\). However, since the alphabets
\(\Sigma_1\) and \(\Sigma_2\) will be clear from the context, we will omit
them. 
Like shuffle, asynchronous product is also associative and commutative, and so we write \(L_1 \parallel \cdots \parallel L_n\). Notice finally that
shuffle and asynchronous product coincide if $\Sigma_1 \cap \Sigma_2 = \emptyset$, but usually differ otherwise. For instance, if \(L_1 = ab^\omega\) and \(L_2 = ab^\omega\), we get \(L_1 \parallel L_2 = ab^\omega\). 

We describe systems as combinations of shuffles and asynchronous products, for instance
we write \(L_1 \parallel (L_2 \between L_3)\). In these expressions we
assume that \(\between\) binds tighter than \(\parallel\), and so
\(L_1 \between L_2 \parallel L_3\) is the language \((L_1 \between L_2)\parallel L_3\),
and not \(L_1 \between (L_2 \parallel L_3)\).

\subsection{Non-atomic networks}
\label{subsec:na-networks}
A non-atomic network is an infinite family of systems parameterized by a number
\(k\).  The \(k\)th element of the family has \(k+1\) components communicating
through a global store by means of read and write actions.  The store is
modeled as an atomic register whose set of possible values is finite.  One
of the \(k+1\) components is the leader, while the other \(k\) are the
contributors. All contributors have exactly the same possible behaviors (they
are copies of the same $\omega$-language), while the leader may behave
differently. The network is called non-atomic because components cannot
atomically execute sequences of actions, only one single read or write.

Formally, we fix a finite set \(\G\) of \emph{global values}.
A {\em read-write alphabet} is any set of the form \(\Act \times \G\), where
\(\Act\) is a set of {\em read} and {\em write (actions)}. We denote a symbol
\((a,g) \in \Act \times \G\) by \(a(g)\) and define \(\G(a_1, \ldots, a_n) = \{
a_i(g) \mid 1\leq i\leq n,\, g\in\G\} \).

We fix two languages \(\D\subseteq\Sigma_{\D}^\omega\) and
\(\C\subseteq\Sigma_{\C}^\omega\), called the \emph{leader} and the
\emph{contributor}, with alphabets \(\Sigma_{\D}=\G(r_d,w_d)\) and
\(\Sigma_{\C}=\G(r_c, w_c)\), respectively, where \(r_d, r_c\)  are called
\emph{reads} and \(w_c, w_d\) are called \emph{writes}.   We write \(\wb\)
(respectively, \(\rb\)) to stand for either \(w_c\) or \(w_d\) (respectively,
\(r_c\) or \(r_d\)). We further assume that
\(\proj_{\{\rb(g),\wb(g)\}}(\D\cup\C)\neq\emptyset\) holds for every 
\(g\in\G\), else the value \(g\) is never used and can be removed from \(\G\). 

Additionally, we fix an $\omega$-language \(\S\), called the {\em store}, 
over \(\Sigma_{\D}\cup\Sigma_{\C}\). It models
the sequences of read and write operations supported 
by an atomic register: a write \(\wb(g)\) writes \(g\) to the register,
while a read \(\rb(g)\) succeeds when the register's current value is \(g\). 
Initially the store is only willing to execute a write. 
Formally \(\S\) is defined as 
\( \Bigl( \sum_{g \in \G} \big(\;\wb(g)\, (\rb(g))^*\;\big) \Bigr)^\omega + \Bigl( \sum_{g \in \G} \big(\;\wb(g)\, (\rb(g))^*\;\big)^* \; \sum_{g \in \G} \big(\;\wb(g)\, (\rb(g))^\omega\;\big) \Bigr) \) and any finite prefix thereof.
Observe that \(\S\) is completely determined by \(\Sigma_{\D}\) and \(\Sigma_{\C}\).
Figure~\ref{fig:example} depicts a store with \(\{1,2,3\}\) as possible values as the language
of a transition system.

\begin{figure}[t]
\centering %
\noindent %
\begin{minipage}[c]{1.75cm}
\begin{tikzpicture}[->,>=stealth',shorten >=1pt,auto,node distance=1.5cm, semithick]
\node[state,initial,initial text = {}, minimum size=2ex] (A) {};
\node[state,minimum size=2ex] (B) [above right of=A] {};
\node[state,minimum size=2ex] (C) [below right of=A] {};
\path[->] (A) edge node[left, font=\scriptsize] {$r_d(1)$} (B)
          (B) edge node[left, font=\scriptsize] {$r_d(2)$} (C)
          (C) edge node[left, font=\scriptsize] {$r_d(3)$} (A);
\end{tikzpicture}
\end{minipage}%
\hspace{\stretch{1}}%
\begin{minipage}[c]{5.5cm}
\begin{tikzpicture}[->,>=stealth',shorten >=1pt,auto,node distance=1.5cm, semithick]
\node[state,initial,initial text = {},minimum size=2ex] (A) {$\#$};
\node[state,minimum size=2ex] (B) [right of=A] {$2$};
\node[state,minimum size=2ex] (C) [above right of=B, xshift=15mm] {$1$};
\node[state,minimum size=2ex] (D) [below right of=B, xshift=15mm] {$3$};

\path[->] (A) edge [bend left=35]  node[above, font=\scriptsize] {$w_\star(1)$} (C) 
              edge node[above, font=\scriptsize] {$w_\star(2)$} (B) 
              edge [bend right=35] node[below, font=\scriptsize] {$w_\star(3)$} (D)
          (B) edge [loop above] node[above, font=\scriptsize] {$rw_\star(2)$} (B)
              edge [bend left=8] node[left,pos=0.7,yshift=1mm, font=\scriptsize]{$w_\star(1)$} (C) 
              edge [bend left=8] node[right,pos=0.5,yshift=1mm, font=\scriptsize]{$w_\star(3)$}(D) 
          (C) edge [loop above] node[left, font=\scriptsize] {$rw_\star(1)$} (C)
              edge [bend left=8] node[right,pos=0.5,yshift=-1mm, font=\scriptsize]{$w_\star(2)$} (B) 
              edge [bend left=8] node[right, font=\scriptsize]{$w_\star(3)$}(D)
          (D) edge [loop below] node[left, font=\scriptsize] {$rw_\star(3)$}(D)
              edge [bend left=8] node[left,yshift=-1mm, font=\scriptsize]{$w_\star(2)$} (B) 
              edge [bend left=8] node[left,pos=0.5, font=\scriptsize]{$w_\star(1)$}(C);
\end{tikzpicture}
\end{minipage}%
\hfill%
\begin{minipage}[c]{3.5cm}
\begin{tikzpicture}[->,>=stealth',shorten >=1pt,auto,node distance=1.5cm, semithick]
\node[state,initial,initial text = {}, minimum size=2ex] (A) {};
\node[state,minimum size=2ex] (B) [above of=A, xshift= 6mm] {};
\node[state,minimum size=2ex] (C) [above of=A, xshift=-6mm] {};
\node[state,minimum size=2ex] (D) [right of=A, yshift=5mm] {};
\node[state,minimum size=2ex] (E) [right of=A, yshift=-5mm] {};
\node[state,minimum size=2ex] (F) [below of=A, xshift = 6mm] {};
\node[state,minimum size=2ex] (G) [below of=A, xshift = -6mm] {};

\path[->] (A) edge node[right,pos=0.65, font=\scriptsize] {$w_c(1)$} (B) 
              edge node[above, font=\scriptsize] {$w_c(2)$} (D) 
              edge node[right,pos=0.65, font=\scriptsize] {$w_c(3)$} (F)
          (B) edge node[above, font=\scriptsize] {$r_c(3)$} (C) 
          (C) edge node[left, font=\scriptsize] {$r_c(1)$} (A)
          (D) edge node[right, font=\scriptsize] {$r_c(1)$} (E) 
          (E) edge node[below, font=\scriptsize] {$r_c(2)$} (A)
          (F) edge node[below, font=\scriptsize] {$r_c(2)$} (G) 
          (G) edge node[left, font=\scriptsize] {$r_c(3)$} (A);
\end{tikzpicture}
\end{minipage}
\caption{Transition systems describing languages $\D$, $\S$, and $\C$.
We write $rw_\star(g)=  \rb(g) \cup \wb(g) = \{ r_c(g), r_d(g)\} \cup \{ w_c(g), w_d(g) \}$.
The transition system for $\S$ is in state $i \in \{1,2,3\}$ when the current
value of the store is $i$.} 
\label{fig:example}
\end{figure}
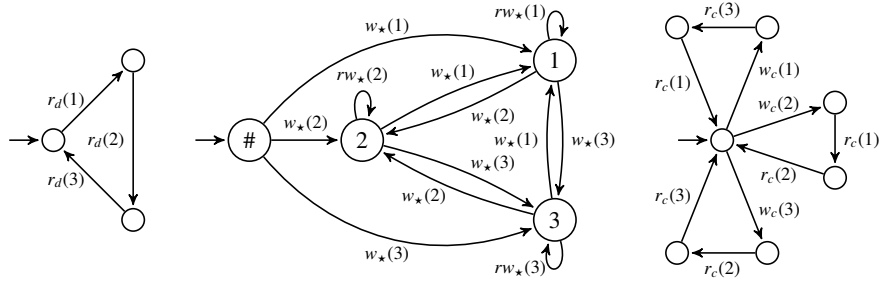
\begin{definition}\label{def:instancenetwork}
Let \(\D \subseteq \Sigma_{\D}^\omega\) and \(\C \subseteq \Sigma_{\C}^\omega\) be a leader and 
a contributor, and let \(k\geq 1\). The \emph{\(k\)-instance of the \((\D,\C)\)-network} 
is the $\omega$-language \ \(\N^{(k)} = (\D \parallel \S \parallel {}\between_{k} \C) \) \ where  
\({}\betweenp{k}\C\) stands for \({}\between_{i=1}^k \C\).
The {\em \((\D,\C)\)-network} \(\N\) is the $\omega$-language \(\N = \bigcup_{k=1}^\infty \N^{(k)}\).
We omit the prefix \((\D,\C)\) when it is clear from the context. 
It follows easily from the properties of shuffle and asynchronous product
that 
$\N = (\D \parallel \S \parallel\, \betweenp{\infty} \C)$,
where \(\betweenp{\infty} \C\) is an abbreviation of \({}\bigcup_{k=1}^\infty \betweenp{k} \C\).
\end{definition}

Next we introduce a notion of {\em compatibility} between a word of the leader
and a multiset of words of the contributor (a multiset because 
several contributors may execute the same sequence of actions).  Intuitively, compatibility means
that all the words can be interleaved into a legal infinite sequence of reads
and writes supported by an atomic register---that is, an infinite sequence
belonging to \(\S\). Formally:

\begin{definition}\label{def:CompRealReachpb}
Let \(u\in \Sigma_{\D}^\omega\), and let \(M=\set{v_1, \ldots, v_k}\) be a multiset of words
over \(\Sigma_{\C}^\omega\) (possibly containing multiple copies of a word).
We say that \(u\) is {\em compatible} with \(M\) if{}f the \(\omega\)-language \((u \parallel \S \parallel {}\between_{i=1}^k v_i)\) is non-empty.
When \(u\) and \(M\) are compatible, there exists a word \(s \in \S\) such that \((u \parallel  s \parallel {}\between_{i=1}^k v_i) \neq \emptyset\).
We call \(s\) a \emph{witness} of compatibility.
\end{definition}

\begin{example}
Consider the network with $\G=\{1,2,3\}$ where the leader, store, and
contributor languages are given by the infinite paths of the
transition systems from Figure~\ref{fig:example}. The only
\(\omega\)-word of $\D$ is $(r_d(1)r_d(2)r_d(3))^\omega$ and the
\(\omega\)-language of $\C$ is
$(w_c(1)r_c(3)r_c(1)+w_c(2)r_c(1)r_c(2)+w_c(3)r_c(2)r_c(3))^\omega$. For
instance, \(\D=(r_d(1)r_d(2)r_d(3))^\omega\) is compatible with the
multiset \(M\) of \(6\) \(\omega\)-words obtained by taking two copies 
of \( (w(1) r(3) r(1))^{\omega}\), \( (w(2) r(1) r(2))^{\omega}\) and \( (w(3) r(2) r(3))^{\omega}\).
The reader may be interested in finding another multiset compatible with $\D$ and containing
only 4 $\omega$-words.
\end{example}

\paragraph{Stuttering property.}
Intuitively, the stuttering property states that if 
we take an $\omega$-word of a network \( \N \) and ``stutter'' reads and writes
of the contributors, going e.g. from $w_d(1) r_c(1)w_c(2)r_d(2) \ldots$
to $w_d(1) r_c(1)r_c(1) w_c(2)w_c(2)w_c(2) r_d(2) \ldots$, the result
is again an $\omega$-word of the network.

Let \(s\in\S\) be a witness of compatibility of \(u \in \Sigma_{\D}^{\omega}\)
and \(M = \set{v_1, \ldots, v_k}\). Pick a set \(I\) of positions (viz. \(I \subseteq \dom(s)\)) such that \( (s)_i \in \Sigma_{\C}\) for each \(i\in I\),
and pick a number \(\ell_i \geq 0\) for every \(i \in I\).  Let \(s'\) be the result of simultaneously 
replacing each \( (s)_i \) by \( (s)_i^{\ell_i+1}\) in \(s\). We have that \(s' \in \S\). 
Now let \(v_s = (s)_{i_1}^{\ell_{i_1}} \cdot (s)_{i_2}^{\ell_{i_2}}\cdots\), where \(i_1 = \min(I)\),  \(i_2 = \min(I\setminus\{i_1\})\), \(\ldots\)
It is easy to see that  \((u \parallel s' \parallel v_s \between {}\between_{i=1}^k v_i) \neq \emptyset\), 
and so \(u\) is compatible with \(M\oplus\{v_s\}\), the multiset consisting of \(M\) and \(v_s\), 
and \(s'\) is a witness of compatibility.

An easy consequence of the stuttering property is the \emph{copycat lemma}~\cite{egm13}.
\begin{lemma}[Copycat Lemma]
\label{lem:extendedmono}
Let \(u \in \Sigma_{\D}^{\omega}\) and let \(M\) be a multiset of words of
\(\Sigma_{\C}^{\omega}\).  If \(u\) is compatible with \(M\), then \(u\) is also compatible with \(M \oplus \{v\}\)
for every \(v \in M\).
\end{lemma}%

\subsection{The Model-checking Problem for Linear-time Properties} 

We consider the model checking problem for linear-time properties, 
that asks, given a network $\N$ and an $\omega$-regular language
$L$, decide whether $\N \parallel L$ is non-empty.
We assume $L$ is given as a B\"uchi automaton $A$ over $\Sigma_{\D}$.
Intuitively, $A$ is a tester that observes the actions of the leader;
we call this the \emph{leader model checking problem}.

We study the complexity of leader model checking for networks in which the 
read-write $\omega$-languages $\D$ and $\C$ of leader and contributor are generated by 
an abstract machine, like a finite-state machine (FSM) or a pushdown machine (PDM). 
(We give formal  definitions later.) 
More precisely, given two classes of machines \texttt{D}, \texttt{C}, 
we study the model checking problem \(\mathtt{MC}(\texttt{D},\texttt{C})\) defined as follows: 
\begin{compactdesc}
\item[{\bf Given}:] machines $D \in \texttt{D}$ and $C \in \texttt{C}$, 
and a B\"uchi automaton $A$ 
\item[{\bf Decide}:] Is $\N_A=(L(A) \parallel L(D) \parallel \S \parallel \betweenp{\infty} L(C))$ non-empty? 
\end{compactdesc}

In the next sections we prove that $\mathtt{MC}(\texttt{FSM},\texttt{FSM})$
and $\mathtt{MC}(\texttt{PDM},\texttt{FSM})$ are NP-complete, 
while $\mathtt{MC}(\texttt{PDM},\texttt{PDM})$ is in NEXPTIME and PSPACE-hard. 

\begin{example}
Consider the instance of the model checking problem where $\D$ and $\C$ are as in Figure~\ref{fig:example},
and $A$ is a B\"uchi automaton recognizing all words over $\Sigma_{\D}$ containing infinitely many 
occurrences of $r_d(1)$. Since $\D$ is compatible with a multiset of words of the contributors, 
$\N_A$ is non-empty. In particular, \(\N_A^{(4)}\neq \emptyset\).
\end{example}

Since \(\Sigma_A=\Sigma_\D\), we can replace \(A\) and \(D\) by a
machine \(A\times D\) with a B\"uchi acceptance condition. The construction
of \(A\times D\) given \(A\) and \(D\) is standard.
In what follows, we assume that \(D\) comes with a B\"uchi acceptance condition and forget about \(A\).

There are two natural variants of the model checking problem, where $\Sigma_A =
\Sigma_{\C}$, i.e., the alphabet of $A$ contains the actions of all
contributors, or $\Sigma_A = \Sigma_{\D} \cup \Sigma_{\C}$.  In both these
variants, the automaton A can be used to simulate atomic networks.  Indeed, if
the language of A consists of all sequences of the form \((w_d() r_c() w_c()
r_d() )^\omega\), and we design the contributors so that they alternate reads
and writes, then the accepting executions are those in which the contributors
read a value from the store and write a new value in an atomic step. So the
complexity of the model-checking problem coincides with the complexity for
atomic networks (undecidable for PDMs and EXPSPACE-complete for FSMs), and we
do not study it further.
 %

%
\makeatletter{}%
%
%
%

\section{$\mathtt{MC}(\texttt{FSM},\texttt{FSM})$ is NP-complete}
\label{sec:fsafsa}

We fix some notations. A finite-state machine (FSM) $(Q,\delta, q_0)$ over $\Sigma$ consists of a
finite set of states $Q$ containing an initial state \(q_0\) and a transition
relation $\delta\subseteq Q\times\Sigma\times Q$.
A word \(v \in \Sigma^{\omega}\) is \emph{accepted} by an
FSM if there exists a
sequence \(q_1q_2\cdots\) of states such that
\( (q_i, (v)_{i+1}, q_{i+1})\in\delta \) for all \(i\geq 0\). We denote by \(q_0 \by{(v)_1} q_1 \by{(v)_2} \cdots\) the \emph{run} accepting \(v\).  
A B\"uchi automaton $(Q,\delta, q_0, F)$ is an FSM $(Q,\delta, q_0)$ 
together with a set $F\subseteq Q$ of accepting
states. An $\omega$-word $v\in\Sigma^\omega$ is accepted by a B\"uchi automaton if there is a run \(q_0 \by{(v)_1} q_1 \by{(v)_2} \cdots\) such that $q_j\in F$  for infinitely many positions \(j\).
The $\omega$-language of a FSM or B\"uchi automaton $A$, denoted by $L(A)$, is the set of $\omega$-words accepted by $A$.

In the rest of the section we show that $\mathtt{MC}(\texttt{FSM},\texttt{FSM})$ is NP-complete. 
Section \ref{subsec:TS} defines the infinite
transition system associated to a (\texttt{FSM},\texttt{FSM})-network. Section \ref{subsec:absTS} introduces
an associated finite abstract transition system. Section \ref{subsec:realiz} states and proves 
a lemma (Lemma \ref{lem:realizability}) characterizing the cycles of the
abstract transition system that, loosely speaking, can be concretized into
infinite executions of the concrete transition system. Membership in NP is then proved using the lemma. 
NP-hardness follows from NP-hardness of
reachability~\cite{egm13}.

\subsection{(\texttt{FSM},\texttt{FSM})-networks: Populations and transition system} 
\label{subsec:TS} 
We fix a B\"uchi automaton $D=(Q_D,\delta_D, q_{0D}, F)$ over \(\Sigma_{\D}\) and an FSM $C=(Q_C, \delta_C, q_{0C})$ 
over \(\Sigma_{\C}\). A {\em configuration} 
is a tuple $(q_D, g, \vec{p})$, where $q_D\in Q_D$, $g \in \G\cup\{\#\}$, and $\vec{p} \colon Q_C \rightarrow \mathbb{N}$ 
assigns to each state of $C$ a natural number. Intuitively, $q_D$ 
is the current state of $D$; $g$ is a value or the special value \(\#\), 
modelling that the store has not been initialized yet, and no process
read before some process writes; finally, $\vec{p}(q)$ is the number of contributors 
currently at state $q \in Q_C$. We call $\vec{p}$ a {\em population} of $Q_C$, and write $|\vec{p}| = \sum_{q \in Q_C} \vec{p}(q)$ for
the {\em size} of $\vec{p}$. Linear combinations
of populations are defined componentwise: for every
state $q \in Q_C$, we have $(k_1 \vec{p_1} + k_2 \vec{p_2})(q) :=
k_1 \vec{p_1}(q) +k_2 \vec{p_2}(q)$. Further, given $q \in Q_C$, we denote by $\vec{q}$ the population
$\vec{q}(q') = 1$ if $q=q'$ and $\vec{q}(q') = 0$ otherwise,
i.e., the population with one contributor in state $q$ and no contributors elsewhere.
A configuration is \emph{accepting} if the state of \(D\) is accepting, that is whenever \(q_D\in F\).
Given a set of populations $\vec{P}$, we define 
$(q_D, g, \vec{P}) := \{ ( q_D, g, \vec{p}) \mid \vec{p} \in \vec{P} \}$. 

The labelled transition system ${\it TS} = (X,T,X_0)$ associated to $\N_A$ is defined
as follows:  
\begin{compactitem}
\item $X$ is the set of all configurations, and $X_0 \subseteq X$ is
 the set of initial configurations, given by $(q_{0D}, \#, \vec{P_0})$, where $\vec{P_0}= \{k\vec{q_{0C}} \mid k \geq 1\}$;
\item $T = T_D \cup T_C$, where
  \begin{compactitem}
  \item $T_D$ is the set of triples $\big( \, ( q_D, g, \vec{p}) \, , \, t \, , \, ( q_D', g', \vec{p}) \, \big)$
  such that $t$ is a transition of $D$, viz. \(t\in\delta_D\), and one of the following conditions holds:
  \begin{inparaenum}[(i)]
  \item $t=(q_D, w_d(g'), q_D')$; or
  \item $t=(q_D, r_d(g), q_D')$, $g = g'$.
  \end{inparaenum}
  \item $T_C$ is the set of triples $\big( \, (q_D, g, \vec{p}) \, , \, t \, , \, (q_D, g', \vec{p}') \, \big)$
  \noindent such that $t\in\delta_C$, and one of the following conditions holds:
  \begin{inparaenum}
  \item[(iii)] $t=(q_C, w_c(g'), q_C')$, $\vec{p} \geq \vec{q_C}$, and 
  $\vec{p}' = \vec{p} - \vec{q_C} + \vec{q'_C}$; or
  \item[(iv)] $t=(q_C, r_c(g), q_C')$, $\vec{p} \geq \vec{q_C}$, $g=g'$, 
  and $\vec{p}' = \vec{p} - \vec{q_C} + \vec{q'_C}$.
  \end{inparaenum}
  \end{compactitem}
  Observe that $|\vec{p}| = |\vec{p}'|$,
  because the total number of contributors  of a population remains constant.
Given configurations $c$ and $c'$, 
we write $c \smash{\by{t}} c'$ if $(c,t,c')\in T$. 
  \end{compactitem}
We introduce a notation important for Lemma \ref{lem:realizability} below.
We define $\Delta(t) :=  \vec{p}' - \vec{p}$. 
Observe that
$\Delta(t)=\vec{0}$ in cases (i) and (ii) above, and $\Delta(t) = - \vec{q_C} + \vec{q_C}'$ 
in cases (iii) and (iv). So $\Delta(t)$ depends only on the transition $t$,
but not on $\vec{p}$.  
\subsection{The abstract transition system} 
\label{subsec:absTS} 
We introduce an {\em abstraction function} $\alpha$ that assigns to a set $\vec{P}$ of populations
the set of states of $Q_C$ populated by $\vec{P}$. We also introduce a
{\em concretization function} $\gamma$ that assigns to a set $Q \subseteq Q_C$ the set of  
all populations $\vec{p}$ that only populate states of $Q$. Formally:
{
\setlength\abovedisplayskip{1pt}
\setlength\belowdisplayskip{1pt}
\begin{align*}
\alpha(\vec{P}) &= \{ q \in Q_C \mid \vec{p}(q)\geq 1 \mbox{ for some $\vec{p}\in \vec{P}$}\}\\
\gamma(Q) &= \{ \vec{p} \mid \vec{p}(q) = 0 \mbox{ for every $q \in Q_C \setminus Q$}\}\enspace .
\end{align*}
}
\noindent It is easy to see that $\alpha$ and $\gamma$ satisfy $\gamma(\alpha(\vec{P})) \supseteq \vec{P}$ 
and $\alpha(\gamma(Q))=Q$, and so $\alpha$ and $\gamma$ form a Galois connection 
(actually, a Galois insertion). 
An {\em abstract configuration} is a tuple $( q_D, g, Q)$, where $q_D \in Q_D$, $g \in \G\cup\{\#\}$, and $Q \subseteq Q_C$. 
We extend $\alpha$ and $\gamma$ to (abstract) configurations in the obvious way.
An abstract configuration is \emph{accepting} when the state of \(D\) is accepting, that is whenever \(q_D\in F\).

Given ${\it TS} = (X,T,X_0)$, we define its  {\em abstraction} 
$\alpha{\it TS} = ({\alpha}X,{\alpha}T,{\alpha}X_0)$ as follows: 
\begin{compactitem}
\item ${\alpha}X = Q_D \times (\G\cup\{\#\}) \times 2^{Q_C}$ is the set of all abstract configurations.
\item \({\alpha}X_0 = (q_{0D}, \#, \alpha(\vec{P_0})) = (q_{0D}, \#, \{q_{0C}\}) \)
is the initial configuration.
\item $(\ (q_D, g, Q),\  t,\ (q_D', g', Q')\ ) \in \alpha{T}$ if{}f
there is $\vec{p}\in\gamma(Q)$ and $\vec{p}'$ such that \\ 
$(q_D,g,\vec{p})\by{t} (q_D',g',\vec{p}')$ and
$Q' = \alpha(\{ \vec{p}' \mid \exists \vec{p} \in\gamma(Q) \colon (q_D, g, \vec{p}) \by{t} (q_D', g', \vec{p}') \})$.
\end{compactitem}
Observe that the number of abstract configurations 
is bounded by $K= \vert{Q_D}\vert \cdot |\G|+1 \cdot 2^{\vert{Q_C}\vert}$. 
Let us point out that our abstract transition system resembles but is different from that of Pnueli et al.\cite{PnueliXZ02}.
We write $a \absby{t} a'$ if \((a,t,a')\in {\alpha}T\). 
The abstraction satisfies the following properties:
\begin{compactitem}
\item[(A)] For each $\omega$-path $c_0 \by{t_1} c_1 \by{t_2} c_2 \cdots$ of ${\it TS}$, 
there exists an \(\omega\)-path
\(a_0 \absby{t_1} a_1 \absby{t_2} a_2 \cdots\) in \(\alpha{\it TS}\) such that 
\(c_i \in \gamma(a_i)\) for all \(i\geq 0\).
\item[(B)] If $(q_D, g, Q) \absby{t} (q_D', g', Q')$, then $Q \subseteq Q'$. \\
To prove this claim, consider two cases:
\begin{compactitem}
\item $t\in \delta_D$. Then $(q_D, g, \vec{p}) \by{t} (q_D', g', \vec{p})$
for every population $\vec{p}$ (because only the leader moves). So 
$(q_D, g, Q) \absby{t} (q_D', g', Q)$. 
\item $t\in \delta_C$. Consider the population 
$\vec{p}= 2 \sum_{q\in Q}  \vec{q} \in \gamma(Q)$. Then 
$( q_D, g, \vec{p}) \by{t} ( q_D, g', \vec{p}')$, where 
$\vec{p}'= \vec{p} - \vec{q_C} + \vec{q_C}'$. But 
then $\vec{p}' \geq \sum_{q\in Q}  \vec{q}$, and so 
$\alpha(\{\vec{p}'\}) \supseteq Q$, which implies 
$( q_D, g, Q) \absby{t} ( q_D, g', Q')$
for some $Q' \supseteq Q$. 
\end{compactitem}
\end{compactitem}
\noindent So in every \(\omega\)-path $a_0 \absby{t_1} a_1 \absby{t_2} a_2 \cdots$ 
of $\alpha{\it TS}$, where $a_i= ( q_{Di}, g_i, Q_i)$, there is an index
$i$ at which the $Q_i$ stabilize, that is,  $Q_i = Q_{i+k}$ holds for every $k \geq 0$.
However, the converse of (A) does not hold: given a path
$a_0 \absby{t_1} a_1 \absby{t_2} a_2 \cdots$ of $\alpha{\it TS}$, there may
be no path $c_0 \by{t_1} c_1 \by{t_2} c_2 \cdots$ in ${\it TS}$ such that
$c_i \in \gamma(a_i)$ for every $i \geq 0$. 
Consider a contributor machine $C$ with two states $q_0, q_1$ and one single transition
$t=(q_0 ,w_c(1), q_1)$. Then \({\alpha}{\it TS}\) contains the infinite path (omitting the 
state of the leader, which plays no role):
{
\setlength\abovedisplayskip{1pt}
\setlength\belowdisplayskip{1pt}
\[(\#, \{q_0\}) \absby{t} (1,\{q_0,q_1\}) \absby{t} (1, \{q_0,q_1\}) \absby{t} (1, \{q_0,q_1\}) \cdots\]
}
\noindent 
However, the transitions of ${\it TS}$ are of the form
\( (1, k_0 \vec{q_0} + k_1 \vec{q_1}) \by{t} (1, (k_0{-}1) \vec{q_0} + (k_1{+}1) \vec{q_1})\), 
and so ${\it TS}$ has no infinite paths.

\subsection{Realizable cycles of the abstract transition system} 
\label{subsec:realiz}
We show that the existence of an infinite accepting path in \({\it TS}\)
reduces to the existence of a certain lasso path in \({\alpha}{\it TS}\).  A
lasso path consists of a stem and a cycle.  
Lemma~\ref{lem:cover} shows
how every abstract finite path (like the stem) has a counterpart in \({\it TS}\).
Lemma~\ref{lem:realizability}
characterizes precisely those cycles in \({\alpha}{\it TS}\) which have an
infinite path counterpart in \({\it TS}\). 

\begin{lemma}
	Let \( (q_D,g,Q) \) be an abstract configuration of \( \alpha{\it TS}\) reachable from 
	\( ( q_{0D}, \#, \alpha(\vec{P_0}))\) (\(={\alpha}X_0\)).
	For every \(\vec{p}\in \gamma(Q)\), there exists \(\vec{\hat{p}}\) such that
	\( ( q_D, g, \vec{\hat{p}})\) is reachable from \( ( q_{0D}, \#, \vec{P_0}) \)
	and \(\vec{\hat{p}} \geq \vec{p}\).
	\label{lem:cover}
\end{lemma}
Lemma~\ref{lem:cover} does not hold for atomic networks. Indeed, consider a contributor
with transitions $q_0 \by{w_c(1)} q_1 \by{r_c(1):w_c(2)} q_2 \by{r_c(2):w_c(3)} q_3$, where 
$r_c(i):w_c(j)$ denotes that the read and the write happen in one single atomic step.
Then we have (omitting the state of the leader, which does not play any r\^ole here):
{
\setlength\abovedisplayskip{1pt}
\setlength\belowdisplayskip{1pt}
\[(\#, \{q_0\}) \absby{w_c(1)} (1, \{q_0, q_1\}) \absby{r_c(1):w_c(2)}
(2, \{q_0, q_1, q_2\})\absby{r_c(2):w_c(3)}
(3, \{q_0, \ldots, q_3\})\enspace .\]
}
\noindent 
Let $\vec{p}$ be the population putting one contributor in each of 
$q_0, \ldots, q_3$. This population belongs to $\gamma(\{q_0, \ldots, q_3\})$
but no configuration \( (3, \vec{\hat{p}})\) with
\(\vec{\hat{p}} > \vec{p}\) is reachable from any population that only puts 
contributors in $q_0$, no matter how many. Indeed, after the first contributor
moves to $q_2$, no further contributor can follow, and so we cannot have contributors
simultaneously in both $q_2$ and $q_3$. On the contrary, in non-atomic networks 
the Copycat Lemma states that what the move by one 
contributor can always be replicated by arbitrarily many.

We proceed to characterized the cycles of the abstract transition system that can be
``concretized''. A {\em cycle} of $\alpha{\it TS}$ is a path 
$a_0 \absby{t_1} a_1 \absby{t_2} a_2 \cdots \absby{t_{n-1}} a_n$
such that $a_n = a_0$. A cycle is {\em realizable} if there is an
infinite path $c_0 \by{t'_1} c_1 \by{t'_2} c_2 \cdots$ of ${\it TS}$ such that 
$c_k \in \gamma(a_{(k \mod n)})$ and $t'_{k+1} = t_{(k+1\mod n)}$ for every $k \geq 0$.

\begin{lemma}\label{lem:realizability}
A cycle $a_0 \absby{t_1} a_1 \absby{t_2} a_2 \cdots \absby{t_{n}} a_n$ of 
$\alpha{\it TS}$ is realizable if{}f $\sum_{i=1}^{n} \Delta(t_i) = \vec{0}$.
\end{lemma}

\begin{theorem}
\label{th:fsafsa}
$\mathtt{MC}(\texttt{FSM},\texttt{FSM})$ is NP-complete.
\end{theorem}
\begin{proof}
NP-hardness
follows from the NP-hardness of reachability \cite{egm13}.
We show membership in NP with the following high-level 
nondeterministic algorithm whose correctness 
relies on Lemmas~\ref{lem:cover} and \ref{lem:realizability}:
\begin{enumerate}
	\item Guess a sequence \(Q_1,\ldots,Q_{\ell}\) of subsets of \(Q_C\) such that \(Q_i \subsetneq Q_{i+1}\) for all \(i\), \(0 < i < \ell\). Note that \(\ell \leq \vert{Q_C}\vert\).
	\item Compute the set \(\mathcal{Q} = Q_D \times (\G\cup\{\#\}) \times \{ \{q_{0C}\} , Q_1,\ldots,Q_{\ell} \} \) of abstract configurations and the set \(\mathcal{T}\) of abstract transitions 
	      between configurations of \(\mathcal{Q}\).
	\item Guess an accepting abstract configuration \(a \in \mathcal{Q}\), that is, an
 \( a = (q_D,g,Q)\) such that \(q_D\) is accepting in \(D\).
	\item Check that \(a\) is reachable from the initial abstract configuration \( ( q_{0D}, \#, \{q_{0C}\}) \) by means of abstract transitions of \(\mathcal{T}\). 
	\item Check that the transition system with $\mathcal{Q}$ and $\mathcal{T}$
as states and transitions contains a cycle \(a_0 \absby{t_1} a_1 \cdots a_{n-1} \absby{t_n} a_n \)
such that \(n\geq 1\), \(a_0=a_n=a\) and \(\sum_{i=1}^n \Delta(t_i) = \vec{0}\).
\end{enumerate}
We show that the algorithm runs in polynomial time. First, because the sequence guessed is no longer 
than \(\vert{Q_C}\vert\), the guess can be done in polynomial time. 
Next, we give a polynomial algorithm for step (5):
\begin{compactitem}
\item Compute an FSA\footnote{A finite-state automaton (FSA) is an FSM which decides languages of finite words. Therefore an FSA is an FSM with a set \(F\) of accepting states.} \(A^{\circlearrowright}_{a}\) over the alphabet
		\(\delta_D\cup\delta_C\) with \(\mathcal{Q}\) as set of states,
		\(\mathcal{T}\) as set of transitions, $a$ as initial state, and  
                \( \{a\} \) as set of final states.
	\item Use the polynomial construction of Seidl \textit{et al.}~\cite{Seidl05} to compute an (existential) Presburger formula \(\Omega\) for the Parikh
		image of \(L(A_{a}^{\circlearrowright})\). 
		The free variables of \(\Omega\) are in one-to-one correspondence with the transitions of \(\delta_D\cup\delta_C\). Denote by \(x_t\)
		the variable corresponding to transition \(t\in \delta_D\cup\delta_C\).
	\item Compute the formula 
		{
\setlength\abovedisplayskip{1pt}
\setlength\belowdisplayskip{1pt}
		\[\textstyle\Omega' = \Omega \; \wedge \;  \bigwedge_{q_c \in Q_c} \Bigl( \sum_{\mathit{tgt}(t) = q_c} x_t = \sum_{\mathit{src}(t) = q_c} x_t \Bigr) \; \wedge \; \sum_{t\in \delta_D\cup \delta_C} x_t > 0\]
		}
\noindent where \(\mathit{tgt}\) and \(\mathit{src}\) returns the target and source states of the transition passed in argument. $\Omega'$ adds to $\Omega$ the realizability condition of Lemma \ref{lem:realizability}.
	\item Check satisfiability of \(\Omega'\). This step requires nondterministic polynomial 
              time because satisfiability of an existential Presburger formula is in NP \cite{Gradel88}.\qed
\end{compactitem}
\end{proof}
 %

%
\makeatletter{}%
%
%
%

\section{$\mathtt{MC}(\texttt{PDM},\texttt{FSM})$ is NP-complete}
\label{sec:pdafsa}

A pushdown system (PDM) $P = (Q, \Gamma, \delta, q_0)$ over $\Sigma$ consists of a finite set $Q$ of states 
including the initial state \(q_0\),
a \emph{stack alphabet} $\Gamma$ including the bottom stack symbol \(\bot\), and a set of \emph{rules}
$\delta \subseteq Q\times \Sigma\times \Gamma \times Q \times (\Gamma\backslash\{\bot\} \cup\{\textsf{pop}\})$ which either push or pop as explained below.
A \emph{PDM-configuration} $qw$ consists of a state $q\in Q$ and a word
$w\in\Gamma^*$ (denoting the stack content).  For $q,q'\in Q$, $a \in
\Sigma$, $\gamma,\gamma'\in \Gamma$, $w,w'\in\Gamma^*$, we say a
PDM-configuration $q'w$ (resp. $q'\gamma'\gamma w$) $a$-follows $q
\gamma w$ if $(q, a, \gamma, q', \textsf{pop})\in\delta$,
(resp. $(q,a,\gamma,q',\gamma') \in \delta$); we write $qw
\by{a} q'w'$ if $q'w'$ $a$-follows $qw$, and call it a {\em transition}.
A \emph{run} $c_0 \by{(v)_1} c_1\by{(v)_2}\ldots$ on a word $v\in\Sigma^\omega$
is a sequence of PDM-configurations such that $c_0 = q_0 \bot$ and 
$c_i \by{(v)_{i+1}} c_{i+1}$ for all \(i\geq 0\).
We write $c\by{*} c'$ if there is a run from $c$ to $c'$. 
The language \(L(P)\) of \(P\) is the set of all words $v\in \Sigma^\omega$
such that $P$ has a run on $v$.

A B\"uchi PDM is a PDM with a set $F\subseteq Q$ of accepting states.
A word is accepted by a B\"uchi PDM if there is a run on the word for which some state in $F$ occurs infinitely often along the
PDM-configurations. The following lemma characterizes accepting runs.

\begin{lemma}{\cite{BEM97}}
\label{lem:pdacycle}
Let $c$ be a configuration.
There is an accepting run starting from $c$ if there are states $q\in Q$, $q_f\in F$,
a stack symbol $\gamma\in\Gamma$ such that
$c \by{*} q\gamma w$ for some $w\in\Gamma^*$
and $q\gamma \by{*} q_f u \by{*} q \gamma w'$ for some $u, w'\in\Gamma^*$.
\end{lemma}

We now show $\mathtt{MC}(\texttt{PDM}, \texttt{FSM})$ is decidable, generalizing the proof from Section~\ref{sec:fsafsa}.
Fix a B\"uchi PDM $P=(Q_D,\Gamma_D, \delta_D, q_{0D}, F)$, and a FSM $C=(Q_C, \delta_C, q_{0C})$. 
A {\em configuration} is a tuple $(q_D, w, g, \vec{p})$, where \(q_D\in Q_D\),
$w\in\Gamma_D^*$ is the stack content, $g \in \G\cup\{\#\}$, and $\vec{p}$ is a population.
Intuitively, $q_D w$ is the PDM-configuration of the leader.
We extend the definitions from Section~\ref{sec:fsafsa} like accepting configuration in the obvious way.

We define a labeled transition system ${\it TS} = (X, T, X_0)$, where $X$ is the set of configurations including the set $X_0 = (q_{0D}, \bot, \#, \vec{P_0})$ of initial configurations,
and the transition relation $T = T_D \cup T_C$, where $T_C$ is as before and $T_D$ is the set of triples
$\big( ( q_D, w, g, \vec{p}), t, ( q_D', w', g', \vec{p})\big)$ 
  such that $t$ is a transition (not a rule) of $D$, and one of the following conditions holds:
  \begin{inparaenum}[(i)]
  \item $t = (q_D w \by{w_d(g')} q_D' w')$; or
  \item $t= (q_Dw \by{r_d(g)} q_D'w')$ and $g = g'$.
  \end{inparaenum}
We define the abstraction $\alpha{\it TS}$ of ${\it TS}$ as the obvious generalization of the abstraction in Section~\ref{sec:fsafsa}.
An accepting path of the (abstract) transition system is an infinite path  with infinitely many accepting (abstract) configurations. As for $\mathtt{MC}(\texttt{FSM}, \texttt{FSM})$, not every accepting path of the abstract admits a concretization, but we find a realizability condition in terms of linear constraints. Here we use again the polynomial construction of 
Seidl \textit{et al.}~\cite{Seidl05} mentioned in the proof of 
Theorem \ref{th:fsafsa}, this time to compute an (existential) Presburger formula
for the Parikh image of a pushdown automaton. 

\begin{theorem}
\label{th:pdafsa}
$\mathtt{MC}(\texttt{PDM},\texttt{FSM})$ is NP-complete.
\end{theorem}

%
%
%
%
 %

%
\makeatletter{}%
%
%
%

\section{$\mathtt{MC}(\texttt{PDM},\texttt{PDM})$ is in NEXPTIME}
\label{sec:pdapda}

We show how to reduce $\mathtt{MC}(\texttt{PDM},\texttt{PDM})$ to
$\mathtt{MC}(\texttt{PDM},\texttt{FSM})$.
We first introduce the notion of {\em effective stack height} of a 
PDM-configuration in a run of a PDM, and define, given a PDM $C$,  
an FSM $C_k$ that simulates all the runs of $C$ of effective stack height $k$.
Then we show that, for $k \in O(n^3)$, where $n$ is the size of $C$,
the language $(L(D) \parallel \S \parallel \betweenp{\infty} L(C))$ is
empty if{}f  $(L(D) \parallel \S \parallel \betweenp{\infty} L(C_k))$
is empty. 

\subsection{A FSM for runs of bounded effective stack height}

Consider a run of a PDM that repeatedly pushes symbol on the
stack. The stack height of the configurations\footnote{
For readability, we write ``configuration'' for ``PDM-configuration.''
}
 is unbounded, but,
intuitively, the PDM only uses the topmost stack symbol during
the run. To account for this we define the notion of effective stack height.

\begin{definition}
  Let $\rho=c_0 \by{(v)_1} c_1 \by{(v)_2} \cdots$ be an infinite run of a
  PDM on \(\omega\)-word \(v\), where $c_i = q_iw_i$. 
  The {\em dark suffix} of $c_i$ in $\rho$, denoted by $\ds{w_i}$, is
  the longest suffix of $w_i$ that is also a proper suffix of $w_{i+k}$ for
  every $k \geq 0$. The {\em active prefix} $\ap{w_i}$ of $w_i$ is the
  prefix satisfying $w_i = \ap{w_i} \cdot \ds{w_i}$.
	The {\em effective stack height} of \(c_i\) in $\rho$ is 
	$\vert{\ap{w_i}}\vert$.
	We say that $\rho$ is {\em effectively $k$-bounded} (or simply \(k\)-bounded
	for the sake of readability) if every configuration of $\rho$ has an
	effective stack height of at most $k$. 
	Further, we say that $\rho$ is {\em bounded} if it is $k$-bounded for some $k
	\in \mathbb{N}$.  Finally, an $\omega$-word of the PDM is {\em $k$-bounded},
	respectively bounded, if it is the word generated by some $k$-bounded,
	respectively bounded, run (other runs for the same word may not be bounded).
\end{definition}

Intuitively, the effective stack height measures the actual memory required by
the PDM to perform its run. 
For example, repeatedly pushing symbols on the stack produces a run with
effective stack height 1.  
Given a position in the run, the elements of the stack that are never popped
are those in the longest common suffix of all subsequent stacks.
The first element of that suffix may be read, therefore only the
longest \emph{proper} suffix is effectively useless, so  
no configuration along an infinite run has effective stack height \(0\). 

\begin{proposition}\label{prop:esh1often}
Every infinite run of a PDM contains infinitely many positions at which the
effective stack height is 1.
\end{proposition}
 \begin{proof}
 Let $p_0w_0 \by{} p_1w_1 \by{} p_2w_2 \by{} \cdots$ be any infinite run. 
 Notice that $|w_i|\geq 1$ for every $i \geq 0$, because otherwise the run would not
 be infinite. Let $X$ be the set of positions of the run defined as:
 $i \in X$ if{}f $|w_i| \leq |w_j|$ for every $j > i$. Observe that
 $X$ is infinite, because the first configuration of minimal stack height,
 say $p_kw_k$ belongs to it, and so does the first configuration of minimal stack height
 of the suffix $p_{k+1}w_{k+1} \by{} \cdots$, etc.
 By construction, the configuration at every position in $X$ has effective stack height $1$.\qed
 \end{proof}

In a $k$-bounded run, whenever
the stack height exceeds $k$, the $k+1$-th stack symbol will never become the top symbol
again, and so it becomes useless. 
So, we can construct a finite-state machine $P_k$ recognizing the words of $L(P)$
accepted by $k$-bounded runs. 

\begin{definition}
  Given a PDM $P = (Q,\Gamma, \delta, q_0)$, the FSM $P_k = (Q_k, \delta_{k}, q_{0k})$,
	called the \emph{\(k\)-restriction of \(P\)}, is defined as follows: (a) $Q_k=Q \times \bigcup_{i=1}^k \Gamma^{i}$ (a state of $P_k$ consists of a state of $P$ and a stack content no longer than $k$);
(b) \(q_{0k} = (q_0, \bot)$; (c) $\delta_{k}$ contains a transition $(q, (w)_{1..k}) \by{a} (q', (w')_{1..k})$ if{}f \(q w \by{a} q' w'\) is a transition (not a rule) of $P$.
\end{definition}

\begin{theorem}
\label{th:kbounded}
Given a PDM $P$, $w$ admits a $k$-bounded run in $P$ if{}f $w \in L(P_{k})$.
\end{theorem}

\subsection{The Reduction Theorem}

We fix a B\"uchi PDM \(D\) and a PDM $C$. 
By Theorem~\ref{th:kbounded}, in order to reduce $\mathtt{MC}(\texttt{PDM},\texttt{PDM})$ to
$\mathtt{MC}(\texttt{PDM},\texttt{FSM})$ it suffices to prove the following
Reduction Theorem:

\begin{theorem}[Reduction Theorem]
  \label{th:redPDM}
	Let $N = 2 |Q_C|^2 |\Gamma_C| + 1$, where $Q_C$ and $\Gamma_C$ are the
	states and stack alphabet of $C$, respectively. Let $C_N$ be the $N$-restriction of $C$.
  We have: 
	{
	\setlength\abovedisplayskip{1pt}
	\setlength\belowdisplayskip{1pt}
	\begin{equation} \label{eq:reduction} \big(L(D) \parallel \S \parallel \betweenp{\infty} L(C)\big) \neq
	\emptyset \quad\text{if{}f}\quad \big(L(D) \parallel \S \parallel 
	\betweenp{\infty} L(C_N)\big) \neq \emptyset \enspace .\tag{$\dagger$} \end{equation} 
	}
  There are PDMs $D$, $C$ for which $\eqref{eq:reduction}$ holds only for $N\in\Omega(|Q_C|^2|\Gamma_C|)$. 
\end{theorem}

Theorems \ref{th:redPDM} and \ref{th:pdafsa}
provide an upper bound for $\mathtt{MC}(\texttt{PDM},\texttt{PDM})$. 
PSPACE-hardness of the reachability problem \cite{egm13} 
gives a lower bound.

\begin{theorem}
  $\mathtt{MC}(\texttt{PDM},\texttt{PDM})$ is in NEXPTIME and
  PSPACE-hard. If the contributor is a one counter machine (with
  zero-test), it is NP-complete.
\end{theorem}

The proof of Theorem~\ref{th:redPDM} is very involved. 
Given a run of $D$ compatible with a finite multiset of runs of $C$,
we construct another run of $D$ compatible with a 
finite multiset of $N$-bounded runs of $C_N$. (Here we
extend compatibility to runs: runs are compatible if the words they accept
are compatible.)

The proof starts with the Distributing lemma, which, loosely speaking, shows how
to replace a run of $C$ by a multiset of ``smaller'' runs of $C$ without the leader
``noticing''. After this preliminary result, 
the first key proof element is the Boundedness Lemma.
Let $\sigma$ be an infinite  run of $D$ compatible with a finite multiset $R$ 
of runs of $C$. The Boundedness Lemma states that, for any number $Z$, the first $Z$ steps of
 $\sigma$ are compatible with a (possibly larger) multiset $R_Z$ of runs of $C_N$. 
Since the size of $R_Z$ may grow with $Z$, this lemma does not
yet prove Theorem~\ref{th:redPDM}: it only shows that $\sigma$ is compatible
with an {\em infinite} multiset of runs of $C_N$. This obstacle is overcome in the 
final step of the proof. We show that, for a sufficiently large $Z$, 
there are indices $i<j$ such that,
not $\sigma$ itself, but the run $(\sigma)_{1..i} \big( (\sigma)_{i+1..j} \big)^\omega$ 
for adequate $i$ and $j$ is compatible with a {\em finite} multiset of runs of $C_N$. Loosely speaking,
this requires to prove not only that the leader can repeat $(\sigma)_{i+1..j}$ infinitely often, but also that the runs executed by the instances of $C_N$ while the leader executes 
$(\sigma)_{i+1..j}$ can be repeated infinitely often.

\smallskip
\noindent\textit{The Distributing Lemma.}
Let $\rho = c_0 \by{a_1} c_1 \by{a_2} c_2 \by{a_3} \cdots$ be a (finite or infinite) run of $C$.
Let $r_i$ be the PDM-rule of $C$ generating the transition $c_{i-1} \by{a_i} c_i$.
Then $\rho$ is completely determined by $c_0$ and the sequence $r_1r_2r_3$ \ldots 
Since $c_0$ is also fixed (for fixed $C$), in the rest of the paper
we also sometimes write $\rho = r_1r_2r_3 \ldots $
This notation allows us to speak of \(\dom(\rho)\), \( (\rho)_k\), \( (\rho)_{i..j}\) and \( (\rho)_{i..\infty}\).

We say that {\em $\rho$ distributes to a multiset $R$ of runs of $C$} if there exists 
an {\em embedding function} $\psi$ that assigns to each run $\rho' \in R$ and to each 
position $i \in \dom(\rho')$ a position $\psi(\rho', i) \in \dom(\rho)$, and 
satisfies the following properties:
\begin{compactitem}
\item \( (\rho')_i = (\rho)_{\psi(\rho',i)}\).
(A rule occurrence in $\rho'$ is matched to another occurrence of the same rule in $\rho$.)
\item $\psi$ is surjective.
(For every position $k \in \dom(\rho)$ there is at least one $\rho' \in R$ and a position $i \in \dom(\rho')$ such that $\psi(\rho',i)=k$, or, informally, $R$ ``covers'' $\rho$.)
\item If $i < j$, then $\psi(\rho', i) < \psi(\rho', j)$.
(So $\psi(\rho',1)\psi(\rho',2)\cdots$ is a scattered subword of $\rho$.)
\end{compactitem}
\begin{example} 
\label{ex:dist}
Let $\rho$ be a run of a PDM $P$. Below are 
two distributions $R$ and $S$ of \(\rho = r_ar_br_br_cr_cr_c\). On the left we have \(R = \{\rho'_1, \rho'_2, \rho'_3\}\), and its embedding function \(\psi\); on the right \(S = \{\sigma'_1, \sigma'_2, \sigma'_3\}\), and 
its function \(\psi'\).

\noindent
\hspace*{\stretch{1}}
\(\begin{array}{r|ccc}
\psi    &  1  &  2  &  3   \\ \hline
\rho'_1 &  1  &  6  &      \\
\rho'_2 &  1  &  2  &  5   \\
\rho'_3 &  1  &  3  &  4
\end{array}\) %
\hspace{\stretch{1}} %
\(\begin{array}{rccccccc}
        &   &   1  &  2  &  3  &  4  &  5  & 6  \\ \hline
\rho    & = &  r_a & r_b & r_b & r_c & r_c & r_c\\
\rho'_1 & = &  r_a &     &     &     &     & r_c\\
\rho'_2 & = &  r_a & r_b &     &     & r_c &    \\
\rho'_3 & = &  r_a &     & r_b & r_c &     &    
\end{array}\) %
\hspace*{\stretch{1}} %
\rule[-1cm]{.3mm}{2cm} %
\hspace*{\stretch{1}} %
\(\begin{array}{r|ccccccc}
\psi'     &  1  &  2  &  3  &  4  \\ \hline
\sigma'_1 &  1  &  4  &     &     \\
\sigma'_2 &  1  &  2  &  4  &  5  \\
\sigma'_3 &  1  &  3  &  5  &  6
\end{array}\) %
\hspace{\stretch{1}} %
\(\begin{array}{rccccccc}
          &   &   1  &  2  &  3  &  4  &  5  & 6  \\ \hline
\rho      & = &  r_a & r_b & r_b & r_c & r_c & r_c\\
\sigma'_1 & = &  r_a &     &     & r_c &     &    \\
\sigma'_2 & = &  r_a & r_b &     & r_c & r_c &    \\
\sigma'_3 & = &  r_a &     & r_b &     & r_c & r_c   
\end{array}\)%
\hfill\hspace{0pt}
\smallskip
\end{example}
\begin{lemma}[Distributing lemma]
  \label{lem:distributing}
  Let $u \in L(D)$, and let $M$ be a multiset
  of words of $L(C)$ compatible with $u$. Let $v \in M$ and let
  $\rho$ an accepting run of $v$ in $C$ that distributes to a multiset
  $R$ of runs of $C$, and let $M_R$ the corresponding multiset of words.
  Then $M \ominus \{v\} \oplus M_R$ is also compatible with $u$.
\end{lemma}

\smallskip
\noindent\textit{The Boundedness Lemma.}
We are interested in distributing a multiset of runs of $C$ into another multiset
with, loosely speaking, ``better'' effective stack height.

Fix a run $\rho$ of $C$ and a distribution $R$ of $\rho$ with
embedding function $\psi$. In
Example~\ref{ex:dist}, $(\rho)_{1..4}$ is distributed
into $(\rho'_1)_{1..1}$, $(\rho'_2)_{1..2}$ and $(\rho'_3)_{1..3}$.  Assume $\rho$
is executed by one contributor. We can replace it by 3 contributors executing
$\rho'_1,\rho'_2,\rho'_3$, without the rest of the network noticing any difference. 
Indeed, the three processes can execute $r_a$ immediately after each other, which for the rest of
the network is equivalent to the old contributor executing one $r_a$. Then we
replace the execution of  \((\rho)_{2..4}\) by \((\rho'_2)_2
(\rho'_3)_{2..3}\). 

We introduce some definitions
allowing us to formally describe such facts. Given $k \in \dom(\rho)$, we denote by $c(\rho,k)$ the
configuration reached by $\rho$ after $k$ steps.
We naturally extend this notation to define $c(\rho, 0)$ as the initial configuration.
We denote by ${\it last}_{\psi}(\rho', i)$ 
the largest position \(k \in\dom(\rho')\) such that $\psi(\rho',k)\leq i$ (similarly if none exists, we fix ${\it last}_\psi(\rho',i) =0$).
Further, we denote by $c_\psi(\rho',k)$ the configuration reached by $\rho'$ after $k$ steps of $\rho$, 
that is, the configuration reached by $\rho'$ after the execution of ${\it last}_{\psi}(\rho', k)$ transitions; formally,
 \(c_\psi(\rho',k) = c(\rho', {\it last}_{\psi}(\rho', k))\).
\begin{example}
\label{ex:dist2}
Let $\rho$, $R$, and $\psi$ as in Example~\ref{ex:dist}. 
Assuming that the PDM $P$ has one single state $p$, stack symbols $\{\bot, \alpha\}$ such
that the three rules \(r_a,r_b\) and \(r_c\) are given by $r_a \colon p\bot \rightarrow p\alpha\bot$, $r_b \colon p\alpha \rightarrow p\alpha\alpha$, 
and $r_c \colon p\alpha \rightarrow p$, then we have $c(\rho,5)=p\alpha\bot$. Further, 
${\it last}_\psi(\rho'_1, 5)=1$, ${\it last}_\psi(\rho'_2, 5)=3$, and ${\it last}_\psi(\rho'_3, 5)=3$. Finally,
$c_\psi(\rho'_1, 5)= p \alpha\bot$, $c_\psi(\rho'_2, 5)=p\alpha\bot$, and $c_\psi(\rho'_3, 5) = p\alpha\bot$. 
\end{example}
Given $Z \in \dom(\rho)$ and $K \in \mathbb{N}$, we say that a distribution $R$ of $\rho$
is {\em $(Z,K)$-bounded} if for every $\rho' \in R$ and for every $i \leq Z$, the effective 
stack height of $c_\psi(\rho',i)$ is bounded by $K$. Further, we say that 
$R$ is {\em synchronized} if for every configuration $c(\rho, i)$ with  effective stack height 1 
and for every $\rho' \in R$, \(c_\psi(\rho', i) = c(\rho,i)\) (same control state and same stack content), and also has effective stack 
height 1.\footnote{
	Notice that the effective stack height of a configuration depends on the run it belongs to, 
	and so $c(\rho, i) = c_\psi(\rho', i)$ does not necessarily imply that they have the same effective stack height.
}
The Boundedness Lemma states that there is a constant $N$, depending only on $C$,
such that for every run $\rho$ of $C$ and for every $Z \in \dom(\rho)$ there is 
a $(Z,N)$-bounded and synchronized distribution $R_Z$ of $\rho$. The key of the proof
is the following lemma.
\begin{lemma}
\label{lem:flattening}
Let $N = 2 |Q_C|^2 |\Gamma_C|+1$. Let $\rho$ be a run of $C$ and $Z \in \dom(\rho)$
be the first position of $\rho$ such that $c(\rho, Z)$ is not $N$-bounded. 
Then there is a $(Z,N)$-bounded and synchronized distribution of $\rho$. 
\end{lemma}
\noindent\textit{Proof sketch.}
We construct a $(Z,N)$-bounded and synchronized distribution $\{\rho_a, \rho_b\}$ of $\rho$.
Let $\alpha_{N+1}\alpha_N \cdots \alpha_1w_0$ be the stack content of $c(\rho,Z)$. 
Define \(\{ \acute{p}_1, \grave{p}_1, \acute{p}_2, \grave{p}_2, \ldots, \acute{p}_N, \grave{p}_N\} \subseteq \dom(\rho) \) such that for each \(i\), \(1\leq i\leq N\) we have
\(c(\rho,\acute{p}_i)\) and \(c(\rho,\grave{p}_i)\)
are the configurations immediately after the symbol $\alpha_i$ in $c(\rho,Z)$
is pushed, respectively popped and such that the stack content of each
configuration between \(\acute{p}_i\) (included) and \(\grave{p}_i\) (excluded)
equals \(w_p \alpha_i \alpha_{i-1} \cdots \alpha_1 w_0\) for some \(w_p \in
\Gamma_C^*\). We get \(c(\rho,\acute{p}_i)	= q_i \alpha_i \alpha_{i-1} \ldots \alpha_0 w_0\) and \(c(\rho,\grave{p}_i)	= q'_i \alpha_{i-1} \ldots \alpha_0 w_0\) for some \(q_i, q'_i\in Q_C\).
Observe that the following holds:\linebreak
\(\acute{p}_1 < \cdots < \acute{p}_{N-1} < \acute{p}_N < Z < \grave{p}_N < \grave{p}_{N-1} < \cdots < \grave{p}_1\).

Since $N = 2 |Q_C|^2 |\Gamma_C|+1$, by the pigeonhole principle we 
find $q,\alpha,q'$ and three indices $1 \leq j_1 < j_2 < j_3 \leq N$ such that 
by letting
$w_1= \alpha_{j_1-1} \cdots \alpha_1$, 
$w_2= \alpha_{j_2-1} \cdots \alpha_{j_1}$ and $w_3 = \alpha_{j_3-1} \cdots \alpha_{j_2}$, we have: 
{
\setlength\abovedisplayskip{1pt}
\setlength\belowdisplayskip{1pt}
\begin{multline*}
	\rho =  (\rho)_{1..\acute{p}_{j_1}} 
\ [q\alpha w_1]\ 
(\rho)_{\acute{p}_{j_1+1}..\acute{p}_{j_2}}
\ [q\alpha w_2 w_1]\ 
(\rho)_{\acute{p}_{j_2+1}..\acute{p}_{j_3}}
\ [q\alpha w_3 w_2 w_1] \\
(\rho)_{\acute{p}_{j_3+1}..\grave{p}_{j_3}}
\ [q'w_3 w_2 w_1]\ 
(\rho)_{\grave{p}_{j_3+1}..\grave{p}_{j_2}}
\ [q' w_2 w_1]\ 
(\rho)_{\grave{p}_{j_2+1}..\grave{p}_{j_1}}
\ [q' w_1]\ 
(\rho)_{\grave{p}_{j_1+1}..\infty}\enspace .
\end{multline*}
}
Here, the notation indicates that we reach configuration $[q\alpha w_1]$ after $(\rho)_{1..\acute{p}_{j_1}}$,
the configuration $[q\alpha w_2w_1]$ after $(\rho_{1..\acute{p}_{j_2}}$, etc.

Now define \(\rho_a\) from \(\rho\) by simultaneously deleting 
\( (\rho)_{\acute{p}_{j_1+1}..\acute{p}_{j_2}}\) and 
\( (\rho)_{\grave{p}_{j_2+1}..\grave{p}_{j_1}}\). We similarly define \(\rho_b\) 
by deleting 
\( (\rho)_{\acute{p}_{j_2+1}..\acute{p}_{j_3}}\) and 
\( (\rho)_{\grave{p}_{j_3+1}..\grave{p}_{j_2}}\).
The following shows that \(\rho_a\) defines a legal run since it is given by
{
\setlength\abovedisplayskip{3pt}
\setlength\belowdisplayskip{3pt}
\[ (\rho)_{1..\acute{p}_{j_1}} 
\ [q\alpha w_1]\ 
(\rho)_{\acute{p}_{j_2+1}..\acute{p}_{j_3}}
\ [q\alpha w_3 w_1] \\
(\rho)_{\acute{p}_{j_3+1}..\grave{p}_{j_3}}
\ [q' w_3 w_1]\ 
(\rho)_{\grave{p}_{j_3+1}..\grave{p}_{j_2}}
\ [q' w_1]\ 
(\rho)_{\grave{p}_{j_1+1}..\infty}\enspace .\]
}
A similar reasoning holds for \(\rho_b\).
Finally, one can show that
$\{\rho_a, \rho_b\}$ is a $(Z,N)$-bounded and synchronized distribution of $\rho$.

\begin{lemma}[Boundedness Lemma]
\label{lem:Z-N-bounded}
Let $N = 2|Q_C|^2 |\Gamma_C|+1$, and let $\rho$ be a run of $C$. For every $Z \in \dom(\rho)$ there is 
an $(Z,N)$-bounded and synchronized distribution $R_Z$ of $\rho$.
\end{lemma}

The proof is by induction on $Z$. The distribution
$\psi_{Z+1},R_{Z+1}$ is obtained from $\psi_{Z}, R_Z$ by distributing
each run $\rho'$ of $R_Z$ to a $(\psi_Z(\rho',Z)+1,N)$-bounded run (applying Lemma~\ref{lem:flattening}).

\smallskip
\noindent\textit{Proof sketch of Theorem~\ref{th:redPDM}.}
Given a run $\sigma$ of $D$ compatible with a finite multiset $M$ of runs of $C$, we
construct another run $\tau$ of $D$, and a multiset $R$ of $N$-bounded runs of $C_N$ such that $\tau$ and $R$ are
compatible as well. We consider only the special  case in which $M$ has one single
element $\rho$ (and one single copy of it). Since $\sigma$ is compatible with
$\rho$, we fix a witness \(\pi \in \S\) such that \(\pi \in \sigma
\between \rho\).  We construct a ``lasso run'' out of $\pi$ of the form \( \lambda_1 [\lambda_2]^{\omega} \). 
It suffices to find two positions in $\pi$ where the content of
the store is the same, the corresponding configurations of the
leader are the same, and similarly for each contributor; the fragment between these two positions can be repeated
(is ``pumpable''). 

Given a position $i$ of $\pi$, let $i_\rho$ and $i_\sigma$ denote the corresponding
positions in $\rho$ and $\sigma$.\footnote{
	Position \(p\) in \(\pi\) defines position \(p_{\sigma}\) in \(\sigma\) such that \( (\sigma)_{1..p_\sigma} = \proj_{\Sigma_{\D}}( (\pi)_{1..p})\), 
	similarly \(p_\rho\) is defined as satisfying \((\rho)_{1..p_\rho} = \proj_{\Sigma_{\C}}((\pi)_{1..p})\).} 
Further, for every $Z$ let $R_Z$ be a $(Z, N)$-bounded 
and synchronized distribution of $\rho$ with embedding function $\psi$ (which exists by the Boundedness Lemma). 
Let $R_Z(i_\rho) = \{ c_\psi(\eta, i_\rho) \mid \eta \in R_Z \}$ denote the multiset of configurations reached by the runs
of $R_Z$ after $i$ steps of $\pi$. 
Using Proposition~\ref{prop:esh1often} and that 
\begin{inparaenum}[(i)]
\item the store has a finite number of values, 
\item $R_Z$ is $(Z, N)$-bounded, and
\item there are only finitely many active prefixes of length at most $N$,  
\end{inparaenum}
we can apply the pigeonhole principle to find a sufficiently large number $Z$ and three positions 
$i < j < k \leq Z$ in $\pi$ satisfying the following properties:
\begin{compactitem}
\item[(1)] The contents of the store at positions $i$ and $k$ of $\pi$ coincide.
\item[(2)] The configurations \(c(\sigma,i_\sigma)\) and \( c(\sigma,k_\sigma)\) of the leader 
have effective stack height 1, same topmost stack symbol and same control state. 
Further, $\sigma$ enters and leaves some accepting state between  $i_\sigma$ and $k_\sigma$.
\item[(3)] The configuration $c(\rho, j_\rho)$ has effective stack height 1.
\item[(4)] For every configuration of $R_Z(i_\rho)$ there is a configuration of $R_Z(k_\rho)$ 
with the same control state and active prefix, and vice versa.
\end{compactitem}
Condition  (4) means that, after removing the dark suffixes, $R_Z(i_\rho)$ and $R_Z(k_\rho)$
contain the same pruned configurations, although possibly a different number of times
(same set, different multisets). If we obtain the same multiset, then 
the fragment of $\pi$ between positions $i$ and $k$ is pumpable by (1) and (2), and we 
are done. Otherwise, we use (3) and the fact that $R_Z$ is synchronized (which had not been used so far) to 
obtain a new distribution in which the multisets coincide. This is achieved by 
adding new runs to $R_Z$.

%
%
%
%
 %

\bibliographystyle{splncs03}

\newpage
\section{Appendix}
\makeatletter{}%
%
%
%

\subsection{Proofs of Section \ref{sec:fsafsa}}

\begin{proof}[of Lemma~\ref{lem:cover}]
	The proof is by induction on the length of the abstract path of \({\alpha}{\it TS}\) from
	\({\alpha}X_0 = ( q_{0D}, \#, \alpha(\vec{P_0}))= ( q_{0D}, \#, \{q_{0C}\})\) to the configuration \( (q_D,g,Q) \).

	The base case corresponds to \((q_D,g,Q) = ( q_{0D}, \#,
	\{q_{0C}\})\). For every $\vec{p} \in \gamma(\{q_{0C}\})$ we can just 
        take $\vec{\hat{p}}=\vec{p}$.

	For the inductive case (in which the abstract path has \(n > 0\) transitions), let 
	\( ( q_{1D}, g_1, Q_1) \absby{t} ( q_D, g, Q)\) be the last transition of the
	path, and assume \(t\in T_C\) (in the case \(t\in T_D\) we have \(Q_1 = Q\) and 
        the result follows immediately from the induction hypothesis).
	Let \(q_C\) and \({q_C}'\) be the source and target state of \(t\), respectively.
	It follows from the definition of \(\absby{}\) that:
	\begin{equation}
		Q = \alpha\left( \, \left\{ \vec{p'} \mid \exists \vec{p}\in \gamma(Q_1) \colon \vec{p} \geq \vec{q_C} \land \vec{p'} = \vec{p} - \vec{q_C} + \vec{q_C}' \right\} \, \right)\enspace .	
		\label{eq:Q}
	\end{equation}
	By (B) and \eqref{eq:Q} we have \(Q = Q_1 \cup \{ {q_C}' \} \). If 
        $q_C' \in Q_1$ then we again get \(Q_1 = Q\), and 
        the result follows from the induction hypothesis. So assume
        $q_C' \notin Q_1$.  Given an arbitrary population \(\vec{p}\in \gamma(Q)\), 
        let \(d = \vec{p}({q_C}')\) and consider the population 
        \(\vec{p_1} = \vec{p} - d \vec{q_C}' + d \vec{q_C}\).
        We have \(\vec{p_1}(q_C') = d - d = 0\), and so \(\vec{p_1} \in \gamma(Q_1)\). By
	induction hypothesis, there exists \(\vec{\hat{p}_1} \in
	\gamma(Q_1)\) such that \(\vec{\hat{p}_1} \geq \vec{p_1}\) and
	\(( q_{1D}, g_1,\vec{\hat{p}_1})\) is reachable from \(X_0\) in \({\it TS}\).
        Since \(\vec{\hat{p}_1} \geq \vec{p_1} \geq d\vec{q_C}\), the mapping
        \(\vec{\hat{p}} = \vec{\hat{p}_1} - d\vec{q_C} + d\vec{q_C}'\) is non-negative, 
        and therefore a population. Now let \(d\) contributors of $\vec{\hat{p}_1}$ 
        execute \(t \in \delta_C \) (which is always possible in a non-atomic network).
        We then have $$ (q_{1D}, g_1, \vec{\hat{p}_1}) \by{t^d} (q_D, g,
	\vec{\hat{p}_1} - d\vec{q_C} + d\vec{q_C}') = (q_D, g, \vec{\hat{p}}) $$
        \noindent  and we are done. \qed
\end{proof}

\begin{proof}[of Lemma~\ref{lem:realizability}]
\noindent ($\Rightarrow$) Assume the cycle is realizable. Then there is an infinite path  
$c_0 \by{t_1} c_1 \by{t_2} c_2 \cdots$ of ${\it TS}$ such that 
\(c_0 \in \gamma(a_0), \ldots, c_{n-1}\in\gamma(a_{n-1})\) and \(c_k \in \gamma(a_{(k\mod n)})\)
for every $k \geq n$, and the transitions match. 
Let $\vec{p}_i$ be the population of $c_i$.
Since the number of contributors of a population 
remains constant across transitions, we have $|\vec{p}_i|= |\vec{p}_0|$ for every $i \geq 0$. Since there are only finitely
many different populations of a given size, by the pigeonhole principle there exist $k_1 < k_2$ such that
$\vec{p}_{k_1 n} = \vec{p}_{k_2 n}$. Since
\begin{gather*}
\vec{p}_{k_2 n} = \vec{p}_{k_1 n} + \sum_{i=(k_1 n) +1}^{k_2 n} \Delta(t_i) \enspace ,
\shortintertext{the sum on the right-hand-side is equal to $\vec{0}$. Since}
\sum_{i=(k_1 n) +1}^{k_2 n} \Delta(t_i) = (k_2 - k_1) \sum_{i=1}^{n} \Delta(t_i)
\end{gather*}
\noindent we get $\sum_{i=1}^{n} \Delta(t_i) = \vec{0}$.

\noindent ($\Leftarrow$) Let $a_i = ( q_{Di}, g_i, Q_i)$ for every $0 \leq i \leq n$. 
By (B) we have $Q_0 \subseteq Q_1 \subseteq \cdots \subseteq Q_n = Q_0$, and so
all the $Q_i$ are equal to $Q_0$. Let $\vec{p_0} = \sum_{q \in Q_0} n\vec{q}$ and let
$\vec{p_i} = \vec{p_0} + \sum_{k=1}^{i-1} \Delta(t_k)$ for every $1 \leq i \leq n$.
Then $\vec{p}_i$ is a population for every $0 \leq i \leq n$, since 
$|\sum_{k=1}^{i-1}\Delta(t_k)| \leq \vec{p_0}$. Moreover
$c_0 \by{t_1} c_1 \by{t_2} c_2 \cdots \by{t_{n}} c_n$ is a path of ${\it TS}$ and, since 
$\sum_{i=1}^{n} \Delta(t_i) = \vec{0}$, we have $c_n = c_0$. 
So $c_0 \by{t_{0}} c_{1} \cdots \by{t_{n}} c_n$ is a cycle of \({\it TS}\) that
can be iterated arbitrarily often, which implies that
$a_0 \absby{t_1} a_1 \absby{t_2} a_2 \cdots \absby{t_{n}} a_n$ is realizable.
\qed
\end{proof}

\begin{proof}[of Theorem~\ref{th:fsafsa}]
In the main text we show that the algorithm runs in nondeterministic polynomial time.
We now prove that it is sound and complete. Then we prove that the problem is NP-hard.

\noindent 
{\bf Soundness.} The  algorithm clearly computes a cycle of \(\alpha{\it TS}\) 
satisfying the assumptions of Lemma~\ref{lem:realizability}. So the cycle is realizable.
Let $c_0 \by{t_{0}} c_{1} \cdots $ be the realization of the cycle, and let 
$c_0 = (q_{D0}, g_0, \vec{p}_0)$. Since $a = a_0$ is reachable from the initial abstract configuration,
by Lemma~\ref{lem:cover} there exists a configuration $c_0'=(q_{D0}, g_0, \vec{\hat{p}}_0)$
reachable from \( ( q_{0D}, \#, \vec{P_0}) \) such that $\vec{\hat{p}}_0 \geq \vec{p}_0$.
So the sequence $(t_0 \ldots t_n)^\omega$ can also be executed from $c_0'$.
Since the cycle visits some accepting state $Q_D$ infinitely often, \({\it TS}\)
has an accepting path.

\noindent 
{\bf Completeness.}
Let $c_0 \by{t_1} c_1 \by{t_2} c_2 \cdots$ be an $\omega$-path of ${\it TS}$ on
 which the B\"uchi automaton accepts. Since the number of contributors in the
populations of the path stays constant, there exist, by the pigeonhole
principle, two positions \(i_1 < i_2\) such that \(c_{i_1} = c_{i_2}\), and they are accepting.  
Clearly, we have \(\sum_{k=i_1}^{i_2} \Delta(t_i) = \vec{0}\), 
hence $c_0 \by{t_1} c_1 \by{t_2} c_2 \cdots c_{i_1-1} \by{t_{i_1}} (c_{i_1}
\by{t_{i_1+1}} c_{i_1+1} \cdots \by{t_{i_2}} c_{i_2})^\omega$ is also an
accepting \(\omega\)-path of {\it TS}.
By (A) the path has a counterpart in the abstraction \(\alpha{\it TS}\), that is, there
exists an \(\omega\)-path $a_0 \absby{t_1} a_1 \absby{t_2} \cdots$ in
\(\alpha{\it TS}\) such that \(c_i \in \gamma(a_i)\) for all \(i\geq 0\). 
Notice that
\begin{inparaenum}[\upshape(\itshape i\upshape)]
\item the sequence of transitions fired along these two paths is the same, and
\item the states of \(D\) and the B\"uchi automaton \(A\) coincide in \(c_i\) and \(a_i\) for all \(i\)
\end{inparaenum}
\noindent and so the abstract path is also accepting.
For each \(i\geq 0\), let $a_i= ( q_{Di}, g_i, Q_i)$. We know
from (B) that \(Q_i \subseteq Q_{i+1}\), and so there is a sequence
$Q_1' \subsetneq Q_2' \subsetneq \ldots \subsetneq Q_\ell'$ such that
for every $i \geq 0$ we have $Q_i = Q_j'$ for some $1 \leq j \leq \ell$. 
It follows that every abstract configuration of the path belongs to the 
set ${\cal Q}$, and every transition to the set ${\cal T}$. 

Let $i$ after which the $Q_i$ stabilize, that is, $Q_i = Q_{i+k}$ holds for every $k \geq 0$. 
Therefore, there exist numbers $i, j$ such that
$a_i = a_{i+j}$ and some abstract configuration between $a_i$ and $a_{i+j}$ is accepting.
So the transition system with ${\cal Q}$ as states and ${\cal T}$ as transitions 
contains a configuration $a$ reachable from the initial abstract configuration, and a 
cycle starting and ending at $a$.

\noindent {\bf Hardness.}
NP-hardness follows from the NP-hardness of the safety problem~\cite{egm13},
which asks given a finite-state machine \(D\) for the leader \(\D\) and \(C\)
for the contributor \(\C\)---both of which are languages of finite
words---whether there exists a word of the \( (\D,\C)\)-network \(\N\) that
ends with an occurrence of \(w_d(\$)\).  We say that \(\N\) is \emph{safe}
if{}f it contains no such word. Remark that, for the safety problem, \(\C\) is
assumed to be prefix-closed, hence every state of \(C\) is accepting.  Also, we
assume without loss of generality that every word of \(\D\) ends with
\(w_d(\$)\).  The reduction goes as follows, given an instance of the safety
problem turn \(\D\) into a \(\omega\)-language by appending it
\(r_d(\$)^{\omega}\). We also turn \(\C\) into a \(\omega\)-language by
appending it \( w_c(\#)^{\omega}\) where \(\# \notin \G\) is the corrupted
value nobody else can read. At the machine level this is done by adding to each
accepting state of \(D\) a selfloop labeled with action \(r_d(\$)\) and
interpreting \(\D\) as a B\"uchi automaton. On the other hand since \(\C\) is
prefix closed we have that all its states are accepting. We turn \(C\) into a
FSM by dropping \(F\)---the set of accepting states---and by adding a self-loop
labelled \(w_c(\#)\) to each state. This concludes the hardness proof.\qed
\end{proof}

\subsection{Proofs of Section \ref{sec:pdafsa}}

\begin{proof}[of Theorem~\ref{th:pdafsa}]
Hardness follows from the NP-hardness for $\mathtt{MC}(\texttt{FSM},\texttt{FSM})$.
The non-deterministic polynomial time algorithm is essentially the same as that of Theorem~\ref{th:fsafsa},
except that we have pushdown systems instead of finite-state systems.
As before, we guess a sequence \(Q_1,\ldots,Q_{\ell}\) of subsets of \(Q_C\) such that \(Q_i \subsetneq Q_{i+1}\) for all \(i\), \(0 < i < \ell \leq \vert{Q_C}\vert\).
We construct a B\"uchi PDM whose states $\mathcal{Q}$ are abstract configurations \( Q_D \times (\G\cup\{\#\}) \times \{ \{q_{0C}\} , Q_1,\ldots,Q_{\ell} \} \),
whose stack alphabet is $\Gamma_D$,
whose initial state is $q_0 = ( q_{0D}, \bot, \#, \set{q_{0C}})$,
whose accepting states are accepting abstract configurations (i.e. where $A$ is accepting),
and whose transitions are defined to mimic $\alpha{\it TS}$.

Next, we guess an abstract configuration $q$ and a stack symbol $\gamma$.
We check if there is a word that takes $q_0\bot$ to $q \gamma w$ for some $w\in\Gamma_D^*$.
This check is equivalent to pushdown reachability and can be performed in polynomial time~\cite{BEM97}.
We construct a PDA\footnote{A pushdown automaton (PDA) is a PDM which decides languages of finite words. We define a PDA as a PDM with a set \(F\) of accepting states.}
 \(P^{\circlearrowright}_{q\gamma}\) over finite words that accepts a word  $ u\in \Sigma^*$ if 
there is a run on $u$ from the starting configuration $q \gamma$ to a configuration $q \gamma w'$ for some $w'\in\Gamma_D^*$ that passes
through an accepting abstract configuration.
The PDA \(P^{\circlearrowright}_{q\gamma}\) can be computed in polynomial time. 

Finally, we check if there is a word accepted by the pushdown automaton whose ``weight'' is $\vec{0}$.
For this check, as before, we compute
an (existential) Presburger formula \(\Omega\) for the Parikh
image of \(L(P^{\circlearrowright}_{q\gamma})\). 
The free variables of \(\Omega\) are in one-to-one correspondence with the transitions of the automaton.
We thus adopt the convention that \(x_t\) denotes
the variable corresponding to transition \(t\in \delta_D\cup\delta_C\).
We compute \(\Omega'\) by adding \(\vert{Q_C}\vert\) variables and
\(\vert{Q_C}\vert\) constraints, one per state in \(q_c\in Q_C\): \(\sum_{\mathrm{tgt}(t) =
q_c} x_t = \sum_{\mathrm{src}(t) = q_c} x_t \) where \(\mathrm{tgt}\) and \(\mathrm{src}\) returns the target and source states of the transition passed in argument.
Add also the constraints \(\sum_{t\in \delta_D\cup \delta_C} x_t > 0\) to prevent
\(\vec{0}\) to be returned as a trivial solution.
Finally, we check satisfiability of \(\Omega'\) and accept if $\Omega'$ is satisfiable. This step is in NP because satisfiability
of an existential Presburger formula is in NP \cite{Gradel88}. 

To see that the algorithm is sound, notice that the algorithm accepts if there is a (pushdown) lasso such that the cyclic part has $\vec{0}$ weight.
For an initial population that is large enough (essentially, cubic in the size of the PDM), we can execute the operations on the path to the lasso
and then execute the cycle to come back to the same configuration as the starting point of the lasso.
This lasso can be pumped infinitely often to produce an accepting run of the B\"uchi PDM.

For completeness, we use Lemma~\ref{lem:pdacycle} to deduce that from an accepting run of the B\"uchi PDM, we can find a lasso-shaped path as defined above.
By a similar pigeonhole argument as that of Lemma~\ref{lem:realizability}, we conclude that we can find a cyclic path whose weight is $\vec{0}$.
\qed
\end{proof}

\subsection{Proofs of Section \ref{sec:pdapda}}

\begin{proof}[of Theorem~\ref{th:kbounded}] 
We first prove that if $v$ admits an effectively $k$-bounded run
$\rho$ in $P$ then $w$ also admits a run in $P_k$. Let $\rho= q_0w_0
\by{(v)_1} q_1w_1 \by{(v)_2} \cdots$,
and let $\ap{i}$, resp. $\ds{i}$, denote the active prefix, resp.\ dark
suffix, of $w_i$. Recall that a state of $P_k$ is a pair $(q, s)$,
where $q$ is a state of $P$ and $s$ is a non-empty stack content of
length at most $k$.
  
  For every $i \geq 0$, we inductively define \((q_i, \ap{i}u_i)\)
  where $u_i$ is a possibly empty prefix of $\ds{i}$ and we show that
  $(q_i, \ap{i}u_i) \by{(v)_{i+1}} (q_{i+1}, \ap{i+1}u_{i+1})$ is a
  transition of $P_k$.

  We define $u_0 = \varepsilon$. Observe that $\ap{0}=\bot$ and $\ds{0}
  = \varepsilon$,
 therefore the initial state of $P_k$ is in the
  desired form. For the definition of $u_{i+1}$, assuming that $u_i$
  is already defined, we consider three cases:
  \begin{compactitem}
  \item The transition $q_iw_i \by{(v)_{i+1}}q_{i+1}w_{i+1}$ pops a symbol $\gamma$. \\
    Then $q_i \, \gamma v \by{(v)_{i+1}} q_{i+1} v$ is a transition of $P$ for every 
    $v$, and so, in particular, $q_i \, \gamma \ap{i+1}u_i \by{(v)_{i+1}} q_{i+1} \, \ap{i+1}u_i$ 
    is a transition of $P$. Moreover,  by the definition of an active prefix, we have 
     $\ap{i} = \gamma \ap{i+1}$ and thus $\ds{i} = \ds{i+1}$ therefore $u_i$ is also a prefix of $\ds{i+1}$.
		By induction hypothesis, $|\ap{i}u_i| \leq k$, which implies $|\ap{i+1}u_i| < k$.
		Setting $u_{i+1}$ to be $u_i$ we thus obtain that \(\vert{\ap{i+1}u_{i+1}}\vert\leq k\) and finally that $(q_i, \ap{i}u_i) \by{(v)_{i+1}} (q_{i+1}, \ap{i+1}u_{i+1})$ is a transition of $P_k$.
  \item The transition $q_iw_i \by{(v)_{i+1}}q_{i+1}w_{i+1}$ pushes a symbol 
    $\gamma$, and $|\ap{i}u_i| < k$. \\
    Then $q_i \ap{i} u_i \by{(v)_{i+1}} q_{i+1} \gamma \ap{i} u_i$ is a transition of $P$.
    Since $|\ap{i}u_i| < k$, we have $|\gamma \ap{i}u_i| \leq  k$, hence 
		\((q_i, \ap{i} u_i) \by{(v)_{i+1}} (q_{i+1}, \gamma \ap{i} u_i)\)
		is also a transition
    of $P_k$. If $\gamma$ is popped later on, then $\ap{i+1} = \gamma \, \ap{i}$; so
    $q_i \, \ap{i}u_i \by{(v)_{i+1}} q_{i+1} \, \ap{i+1}u_i$ is a transition of $P$,
    and we set $u_{i+1}$ to $u_i$.  If $\gamma$ is never popped, then
		$\ap{i+1} = \gamma$, and we let $u_{i+1}$ to be $\ap{i} u_i$. 
		In both cases, we find that \(\vert{\ap{i+1} u_{i+1}}\vert \leq k\) and hence that 
		\((q_i, \ap{i} u_i) \by{(v)_{i+1}} (q_{i+1}, \ap{i+1} u_{i+1})\)
		is a transition of \(P_k\).
  \item The transition $q_iw_i \by{(v)_{i+1}}q_{i+1}w_{i+1}$ pushes a symbol 
    $\gamma$, and $|\ap{i}u_i|= k$. \\
    Then $q_i \ap{i}u_i \by{(v)_{i+1}} q_{i+1} \gamma \ap{i} u_i$ is a transition of $P$.
		Observe that since $|\ap{i}u_i|= k$ we have \(\vert{\gamma \ap{i} u_i}\vert = k+1\).
		First we show that $|u_i| > 0$,  if $|u_i| = 0$, then \(\vert{\ap{i}}\vert=k\), and more importantly
                $\ds{i}$ is the largest proper suffix of all the $(w_j)_{j\geq i}$, and since $w_i$ is a proper suffix of $w_{i+1}$,
                $\ds{i}$ is also the largest proper suffix of all the $(w_j)_{j\geq {i+1}}$, therefore $\gamma \ap{i} = \ap{i+1}$,
                so $|\ap{i+1}| = k+1$
                contradicting the hypothesis that the run is effectively $k$-bounded. 

    We can therefore write $u_i = u_i' \gamma'$. Since $|\gamma\ap{i}u_i|= k+1$, 
		$(q_i, \ap{i} u_i) \by{(v)_{i+1}} (q_{i+1}, (\gamma \ap{i} u_i)_{1..k} )$ is a transition of $P_k$.
    If $\gamma$ is popped later on, then $\ap{i+1} = \gamma \, \ap{i}$ and  
    $u_{i+1} = u_i'$. If $\gamma$ is never popped, then $\ap{i+1} = \gamma$, and
		$u_{i+1} = (\ap{i} u_i)_{1..k-1}$. 
		In both cases we conclude that 
		\(\vert{\ap{i+1} u_{i+1}}\vert \leq k\), hence that 
		\((q_i, \ap{i} u_i) \by{(v)_{i+1}} (q_{i+1}, \ap{i+1} u_{i+1})\)
		is a transition of \(P_k\).
  \end{compactitem}

  Now we show that if $v$ admits a run in $P_k$, then it 
	admits an effectively $k$-bounded run $\rho'$ in $P$. Let 
  $\rho = (q_0, w_0) \by{(v)_1} (q_1, w_1) \by{(v)_2} \cdots$ be a run of 
  $P_k$ for $v$ such that $|w_i|\leq k$ for every $i \geq 0$. 
  We inductively construct $w_0', w_1', \ldots$ such that  
  $\rho' = (q_0, w_0w_0') \by{(v)_1} (q_1, w_1w_1') \by{(v)_2} \cdots$ is a run of $P$
  satisfying the following invariant: 
  \begin{equation}
  \label{eq:inv}
	|w_{i+1}w'_{i+1}| - |w_iw'_i| \geq |w_{i+1}| - |w_i|, \quad\text{for all } i\geq 0\enspace .
  \end{equation}

  We start by defining $w'_0 = \varepsilon$, which trivially satisfies~\eqref{eq:inv}. 
  Assume $(q_0, w_0w_0') \by{(v)_1} \cdots \by{(v)_{i+1}} (q_i, w_iw_i')$ is a run of $P$
  satisfying~\eqref{eq:inv}, and consider the transition 
  $(q_i, w_i) \by{(v)_{i+1}} (q_{i+1}, w_{i+1})$ of $P_k$.
  By the definition of the transitions of $P_k$, there are two possible cases: 
  \begin{compactitem}
		\item $q_i w_i \by{(v)_{i+1}} q_{i+1} w_{i+1}$ is a transition of $P$. \\
  Then $q_i w_iw_i' \by{(v)_{i+1}} q_{i+1} w_{i+1}w'_i$ is also a transition
  of $P$, and we can take $w'_{i+1}$ to be $w'_i$, and \eqref{eq:inv} is satisfied as
$|w_{i+1}w'_{i+1}| - |w_iw'_i| = |w_{i+1}| - |w_i|$
\item $q_i w_i \by{(v)_{i+1}} q_{i+1} w_{i+1} \gamma$ is a transition of $P$. \\
  Then $|w_i|=|w_{i+1}|=k$, and 
  $q_i w_iw_i' \by{(v)_{i+1}} q_{i+1} w_{i+1}\gamma w'_i$ is a transition
  of $P$. So setting $w'_{i+1}$ to $\gamma w'_i$ satisfies~\eqref{eq:inv} as
$|w_{i+1}w'_{i+1}| - |w_iw'_i| = |w_{i+1}| - |w_i|  +1$
  \end{compactitem}

  \noindent The induction is concluded, now we explain the meaning of
  equation \eqref{eq:inv}.  First remark that performing a telescope
  sum, we obtain that for any $i,j > 0$, \(\vert{w_{i+j}w'_{i+j}}\vert
  - \vert{w_iw'_i}\vert \geq \vert{w_{i+j}}\vert -
  \vert{w_i}\vert\). Since $|w_{i+j}| \leq k$ and $|w_i| \geq 1$, we
  obtain $\vert{w_{i+j}w'_{i+j}}\vert - \vert{w_iw'_i}\vert \geq 1-k$.
  Informally it means that the number of symbols in the stack at any
  position after $i$ can't be much smaller (much meaning $k$) than at
  position $i$. Thus, at every position $i$, we never eventually pop
  the $k$ top symbols of the stack at that position, as this would
  yield a configuration after $i$ whose stack would be too small and
  contradict the inequality. Therefore the run $\rho'$ is effectively
  $k$-bounded.
\qed
\end{proof}

\begin{proof}[of Lemma~\ref{lem:distributing}, the Distributing Lemma]
Since $u$ is compatible with $M$, there exists a witness $s \in
L(\mathcal{S})$ such that $s \in (u \parallel \betweenp{w\in M}
w)$. Since $\Sigma_{\mathcal{C}} \cap \Sigma_{\mathcal{D}} =
\emptyset$, we have  $(u \parallel \betweenp{w\in M}
w)=(u \between \betweenp{w\in M}w)$, and so $s \in (u \between \betweenp{w\in M}w)$. Therefore there
exists an interleaving function, i.e. a bijection $\mathcal{I} :
\coprod_{w \in u \oplus M} \dom(w) \to \dom(s)$, that assigns to each
position in each word in $u\oplus M$ a corresponding position in $s$
with the same action. Further, the interleaving function satisfies $i < j \in \dom(w)$ if{}f $\mathcal{I}(w,i) < \mathcal{I}(w,j)$.

For example, if $u = w_d(1)$ and $M = \{ w_1, w_2\}$,
where $w_1 = r_c(1) w_c(2) r_c(1)$ and $w_2 = r_c(2)w_c(1)$, then we can take
$s = w_d(1) r_c(1) w_c(2) r_c(2) w_c(1) r_c(1)$, with 
$I(w_1, 1) = 2$, $I(w_1, 2)=3$, $I(w_1, 3)=5$, $I(w_2,1)=4$, $I(w_2,2)=5$.

We have to show that, given $v \in M$, a run $\rho$ of $C$ accepting a word $v$, and a distribution $R$ of $ \rho$ accepting a multiset $M_R$ of words, then $M \ominus \{v\} \oplus M_R$ is compatible with $u$.

Let $\psi$ be the embedding function of the distribution $R$.
We construct a word $s'$ witnessing that $u$ and $M \ominus \{v\} \oplus M_R$ are compatible. The word $s'$ is a stuttering of $s$, that is,
it is obtained from $s$ by repeating some letters of $s$; since, by
definition of the store, $\S$ is closed under stuttering, we have $s' \in \S$. Let $s = a_1a_2 \ldots$, and let $i$ be a position  of
$s$ such that $\mathcal{I}(v,j)=i$ for some $j \in \dom(v)$ (so, loosely speaking, position $j$ in the interleaving $s$ comes
from the word $v\in M$). Further, let $k$ be the number of runs in $R$ such that some position in them is mapped to position $j$ by the embedding function $\psi$ (intuitively, $k$ is the number of runs in $R$
executing the action at position $j$. Then we replace $a_i$ by $a_i^k$ (that is, by the word $a_i \ldots a_i$ of length $k$). 

For example, if we distribute $w_1$ above to $\{ r_c(1) w_c(2), w_c(2)r_c(1), w_c(2) \}$, then we get $s' = w_d(1) r_c(1) (w_c(2) w_c(2)w_c(2)) r_c(2) w_c(1) r_c(1)$.

We clearly have a one-to-one correspondence between positions in $s'$
and positions in $u \oplus M \ominus v \oplus M_R $.\qed
\end{proof}

\paragraph{Proof of the Boundedness Lemma.}

Before proving Lemma~\ref{lem:flattening} and the Boundedness lemma we give an example of two distributions 
of a finite run that decrease the effective stack height, one of them
being moreover synchronized.

\begin{example}
\label{ex:dist3}
Consider the two distributions $R$ and $S$ of $\rho=r_a r_b r_b r_c r_c r_c$ in
Example~\ref{ex:dist}.
Further assume that the PDM $P$ has one single state $p$, stack symbols $\{\bot, \alpha\}$ such
that the three rules \(r_a,r_b\) and \(r_c\) are given by $r_a \colon p\bot \rightarrow p\alpha\bot$, $r_b \colon p\alpha \rightarrow p\alpha\alpha$, 
and $r_c \colon p\alpha \rightarrow p$.
Figure~\ref{fig:moredist} graphically depicts the stack contents of the configurations of the runs (the control
state is always $p$), and their respective effective stack heights.
\begin{figure}[htb]
$$\begin{array}{r|cccccccccc}
               &    &  0       &  1                &  2                   &  3                         &  4                    &  5             &  6    \\ \hline
\rho \colon    & \; & \; \bot \;  &  \; \alpha \bot \;   & \; \alpha\alpha\bot \;  & \; \alpha\alpha\alpha\bot \;  & \; \alpha\alpha\bot \;   & \; \alpha\bot \;  & \; \bot \; \\
{\it e.s.h.}   &    &  1       &  2                &  3                   &  4                         &  3                    &  2             &  1    \\ \hline
\rho'_1 \colon &    &  \bot    &  \alpha\bot       &                      &                            &                       &                & \bot \\
{\it e.s.h.}   &    &  1       &  2                &                      &                            &                       &                & 1     \\ \hline
\rho'_2 \colon &    &  \bot    &  \alpha\bot       &  \alpha\alpha\bot    &                            &                       &  \alpha\bot    &       \\
{\it e.s.h.}   &    &  1       &  1                &  2                   &                            &                       &  1             &       \\ \hline
\rho'_3 \colon &    &  \bot    &  \alpha\bot       &                      &  \alpha\alpha\bot          &  \alpha\bot           &                &       \\
{\it e.s.h.}   &    &  1       &  1                &                      &  2                         &   1                   &                &       \\ \hline
\end{array}
\qquad 
\begin{array}{r|cccccccccc}
                 &    &  0       &  1                &  2                   &  3                         &  4                    &  5             &  6    \\ \hline
\rho \colon      & \; & \; \bot \;  &  \; \alpha \bot \;   & \; \alpha\alpha\bot \;  & \; \alpha\alpha\alpha\bot \;  & \; \alpha\alpha\bot \;   & \; \alpha\bot \;  & \; \bot \; \\
{\it e.s.h.}     &    &  1       &  2                &  3                   &  4                         &  3                    &  2             &  1    \\ \hline
\sigma'_1 \colon &    &  \bot    &  \alpha\bot       &                      &                            &  \bot                 &                &       \\
{\it e.s.h.}     &    &  1       &  2                &                      &                            &  1                    &                &       \\ \hline
\sigma'_2 \colon &    &  \bot    &  \alpha\bot       &  \alpha\alpha\bot    &                            &  \alpha\bot           &  \bot          &       \\
{\it e.s.h.}     &    &  1       &  2                &  3                   &                            &  2                    &  1             &       \\ \hline
\sigma'_3 \colon &    &  \bot    &  \alpha\bot       &                      &  \alpha\alpha\bot          &                       & \alpha\bot     & \bot  \\
{\it e.s.h.}     &    &  1       &  2                &                      &  3                         &                       & 2              & 1     \\ \hline
\end{array}
$$
\caption{Configurations and effective stack heights of the distributions of Example~\ref{ex:dist}.\label{fig:moredist}}
\end{figure}

We observe that $\rho$ is effectively $4$-bounded. The distribution $R$ is $(Z,2)$-bounded for every $1 \leq Z \leq 6$,
 because the configurations $c_\psi(\rho'_j, i)$ have effective stack height
at most $2$ for every $1 \leq j \leq 3$ and every $1 \leq i \leq 6$. The distribution is
not synchronized. Indeed, the configuration $c(\rho, 6) = p \bot$ has effective
stack height 1, but  $c_\psi(\rho'_2, 6) = p \alpha \bot \neq c(\rho, 6)$. The distribution $S$ is 
$(Z,3)$-bounded for every $1 \leq Z \leq 6$ and synchronized. Remark that in each of $\{\sigma'_i\}_{i=1,2,3}$ at positions \({\it last}_\psi(\sigma'_i,0)\) and \({\it last}_\psi(\sigma'_i,6)\) (the only two positions at which $\rho$ has effective stack height $1$), the stack content is $\bot$ thus effective stack height is 1.
\end{example}

\begin{proof}[of Lemma~\ref{lem:flattening}]
For convenience, when we want to denote that, say, 
in a run \(\rho\) the configurations reached after $(\rho)_{1..i}$ and $(\rho)_{1..j}$
are $c$ and $c'$, we write $\rho = (\rho)_{1..i}\ [c]\ (\rho)_{i+1..j}\ [c']\ (\rho)_{j+1..\infty}$.

We construct a $(Z,N)$-bounded and synchronized distribution $\{\rho_a, \rho_b\}$ of $\rho$.
Let $\alpha_{N+1}\alpha_N \cdots \alpha_1w_0$ be the stack content of $c(\rho,Z)$. 
Define \(\{ \acute{p}_1, \grave{p}_1, \acute{p}_2, \grave{p}_2, \ldots, \acute{p}_N, \grave{p}_N\} \subseteq \dom(\rho) \) such that for each \(i\), \(1\leq i\leq N\) we have
\(c(\rho,\acute{p}_i)\) and \(c(\rho,\grave{p}_i)\)
are the configurations immediately after the symbol $\alpha_i$ in $c(\rho,Z)$
is pushed, respectively popped and such that the stack content of each
configuration between \(\acute{p}_i\) (included) and \(\grave{p}_i\) (excluded)
equals \(w_p \alpha_i \alpha_{i-1} \cdots \alpha_1 w_0\) for some \(w_p \in
\Gamma_C^*\). We get \(c(\rho,\acute{p}_i)	= q_i \alpha_i \alpha_{i-1} \ldots \alpha_0 w_0\) and \(c(\rho,\grave{p}_i)	= q'_i \alpha_{i-1} \ldots \alpha_0 w_0\) for some \(q_i, q'_i\in Q_C\).
Observe that the following holds:\linebreak
\(\acute{p}_1 < \cdots < \acute{p}_{N-1} < \acute{p}_N < Z < \grave{p}_N < \grave{p}_{N-1} < \cdots < \grave{p}_1\).

Since $N = 2 |Q_C|^2 |\Gamma_C|+1$, by the pigeonhole principle we 
find $q,\alpha,q'$ and three indices $1 \leq j_1 < j_2 < j_3 \leq N$ such that 
by letting
$w_1= \alpha_{j_1-1} \cdots \alpha_1$, 
$w_2= \alpha_{j_2-1} \cdots \alpha_{j_1}$ and $w_3 = \alpha_{j_3-1} \cdots \alpha_{j_2}$, we have: 
\begin{multline*}
	\rho =  (\rho)_{1..\acute{p}_{j_1}} 
\ [q\alpha w_1]\ 
(\rho)_{\acute{p}_{j_1+1}..\acute{p}_{j_2}}
\ [q\alpha w_2 w_1]\ 
(\rho)_{\acute{p}_{j_2+1}..\acute{p}_{j_3}}
\ [q\alpha w_3 w_2 w_1] \\
(\rho)_{\acute{p}_{j_3+1}..\grave{p}_{j_3}}
\ [q'w_3 w_2 w_1]\ 
(\rho)_{\grave{p}_{j_3+1}..\grave{p}_{j_2}}
\ [q' w_2 w_1]\ 
(\rho)_{\grave{p}_{j_2+1}..\grave{p}_{j_1}}
\ [q' w_1]\ 
(\rho)_{\grave{p}_{j_1+1}..\infty}\enspace .
\end{multline*}
Now define \(\rho_a\) from \(\rho\) by simultaneously deleting 
\( (\rho)_{\acute{p}_{j_1+1}..\acute{p}_{j_2}}\) and 
\( (\rho)_{\grave{p}_{j_2+1}..\grave{p}_{j_1}}\). We similarly define \(\rho_b\) 
by deleting 
\( (\rho)_{\acute{p}_{j_2+1}..\acute{p}_{j_3}}\) and 
\( (\rho)_{\grave{p}_{j_3+1}..\grave{p}_{j_2}}\).
The following shows that \(\rho_a\) defines a legal run since it is given by
\[ (\rho)_{1..\acute{p}_{j_1}} 
\ [q\alpha w_1]\ 
(\rho)_{\acute{p}_{j_2+1}..\acute{p}_{j_3}}
\ [q\alpha w_3 w_1] \\
(\rho)_{\acute{p}_{j_3+1}..\grave{p}_{j_3}}
\ [q' w_3 w_1]\ 
(\rho)_{\grave{p}_{j_3+1}..\grave{p}_{j_2}}
\ [q' w_1]\ 
(\rho)_{\grave{p}_{j_1+1}..\infty}\enspace .\]
A similar reasoning holds for \(\rho_b\).
We conclude by proving two claims.

\paragraph{\(\{\rho_a,\rho_b\}\) is a distribution of \(\rho\).} 
The embedding function $\psi$ for $\rho_a$ (again, the case of $\rho_b$ is analogous) is given by
\[\psi(\rho_a, i) =
\begin{cases}
	i & \mbox{ for $1 \leq i \leq \acute{p}_{j_1}$} \\
	i + (\acute{p}_{j_2} - \acute{p}_{j_1}) &  \mbox{ for $\acute{p}_{j_1}+1 \leq i \leq \grave{p}_{j_2} -(\acute{p}_{j_2} - \acute{p}_{j_1})$}\\
i + (\grave{p}_{j_1} - \grave{p}_{j_2})+ (\acute{p}_{j_2} - \acute{p}_{j_1}) & \mbox{ for $\grave{p}_{j_2} -(\acute{p}_{j_2} - \acute{p}_{j_1})+1 \leq i$ }
\end{cases}\]

\paragraph{$\{\rho_a, \rho_b\}$ is a $(Z,N)$-bounded and synchronized distribution of $\rho$.}
Since the effective stack height of every configuration of $\rho_a$ (resp. $\rho_b$) up to 
position \({\it last}_{\psi}(\rho_a, Z)\) (resp. \({\it last}_{\psi}(\rho_b, Z)\)) is at most $N$, the 
distribution is $(Z,N)$-bounded. Finally, observe that we have 
$c(\rho, i) = c_\psi(\rho_a, i) = c_\psi(\rho_b, i)$ for every $i \leq \acute{p}_{j_1}$ and every 
$i \geq \grave{p}_{j_1}$. Since all configurations of $\rho$ of effective stack height 1 are in these two areas,
the distribution is synchronized.\qed
\end{proof}

In order to prove the Boundedness Lemma (Lemma~\ref{lem:Z-N-bounded}), we introduce a definition that allows us to ``nest'' distributions (that is, to distribute a run into several runs, and then 
distribute one of these runs again into several runs), while preserving the properties of synchronization and boundedness.  

\begin{definition}
  Let $R, \psi$ be a distribution of $\rho$. Let $\rho' \in R$, and let
  $R', \psi'$ be a distribution of $\rho'$.  The
  {\em composition} of $R,\psi$ and $R',\psi'$ is the distribution 
  $R\ominus\{\rho'\}\oplus R', \psi''$ of $\rho$, where the embedding function $\psi''$ is 
  defined as follows:
  \begin{compactitem}
  \item $\psi''(r, i) = \psi(r,i)$ for every $r \in R\ominus\{\rho'\}$, and
  \item $\psi''(r,i) = \psi(\rho',\psi_{\rho'}(r, i))$ for every $r \in R'$.
  \end{compactitem}
\end{definition}

The following lemma proves that the composition of distributions is not ill-defined, that is, 
that the composition of distributions is indeed a distribution of $\rho$. 

\begin{lemma}
\label{lem:nesting}
  Let $R, \psi$ be a distribution of $\rho$. Let $\rho' \in R$ and let
  $R', \psi'$ be a distribution of $\rho'$. The composition $R\ominus
  \{\rho'\} \oplus R', \psi''$ is a distribution of
  $\rho$. 
\end{lemma}

\begin{proof}
  We need to show that $\psi''$ satisfies the three properties of an embedding function. 

  \begin{itemize}
  \item $(\rho'')_{i} = (\rho)_{\psi''(\rho'', i)}$.\\
          If $\rho'' \in R\ominus \{\rho'\}$, then,  as $\psi$ is a
		distribution of $\rho$, we have $ (\rho'')_{i} = (\rho)_{\psi(\rho'', i)}$.
           By definition of $\psi''$, we get $\psi''(\rho'', i)= \psi(\rho'',i)$.
           If $\rho'' \in R'$, then, since $\psi'$ is a distribution,
		$(\rho'')_i = (\rho')_{\psi'(\rho'', i)}$.  Since $\psi$
		is a distribution $(\rho')_{j} = (\rho)_{\psi(\rho', j)}$. Taking
		$j= \psi'(\rho'', i)$, we get $(\rho'')_{i} =
		(\rho')_{\psi'(\rho'', i)} = (\rho)_{\psi(\rho',\psi'(\rho'', i))} = (\rho)_{\psi''(\rho'', i)}$.
           So, for every $\rho'' \in R\ominus \{\rho'\}
		\oplus R'$, we finally obtain $(\rho'')_{i} = (\rho)_{\psi''(\rho'', i)}$.
  \item Surjectivity.\\
    If $k\in \dom(\rho)$, we first exploit the surjectivity of
    $\psi$: either there exists $\rho''\in R\ominus \{\rho\}$, and
    some $i \in \dom(\rho'')$ such that $\psi(\rho'', i) = k$ (which
    means that $\psi''(\rho'', i) = k$) or there is
    some $j \in \dom(\rho')$ such that $\psi(\rho', j) = k$. In the
    latter case, we then exploit the fact that $\psi_{\rho'}$ is a
    distribution of $\rho'$, and deduce that there exists $\rho'' \in
    R'$ and $i\in \dom (\rho'')$ such that $\psi'(\rho'', i) = j$;
    hence we have $\psi(\rho', \psi'(\rho'', i)) = k$, and so
    $\psi''(\rho'', i) = k$.

  \item Monotonicity. \\
    For every $\rho''\in R\ominus \{\rho'\}$, from the monotonicity of $\psi$ we
    obtain that $\psi''(\rho'',i) < \psi''(\rho'',j)$ for every $i < j$ . If $\rho'' \in R'$, first we
    derive from the monotonicity of $\psi'$ that $\psi_{\rho'}(\rho'',i) < \psi_{\rho'}(\rho'',j)$ 
    holds for every $i < j$. Then, by monotonicity of $\psi$, we obtain $\psi(\rho', \psi'(\rho'',i)) <
    \psi(\rho', \psi'(\rho'',j))$, and so $\psi''(\rho'',i) < \psi''(\rho'',j)$.\qed
  \end{itemize}
\end{proof}

\begin{lemma}
  Let $\sigma$ be a run of $D$, and let $M\oplus \{\rho\}$ be a multiset of runs of $C$ compatible with 
  $\sigma$. Let $R, \psi$ be a $(Z,N)-$bounded synchronized
  distribution of $\rho$. For every $\rho' \in R$, let $R_{\rho'},
  \psi_{\rho'}$ be a $(\psi(\rho',Z)+1, N)-$bounded synchronized
  distribution of $\rho'$. Then $\bigoplus_{\rho'\in R}R_{\rho'}$ is a
  $(Z+1,N)-$bounded synchronized distribution of $\rho$.
\end{lemma}

\begin{proof}
    By repeated application of Lemma~\ref{lem:nesting}, $\rho$ can be 
    distributed to $\bigoplus_{\rho'\in R}R_{\rho'}$. Let $\Psi$ be the
    corresponding embedding function, obtained also by repeated 
    application of Lemma~\ref{lem:nesting}. We have to prove 
    that $\bigoplus_{\rho'\in R}R_{\rho'}, \Psi$ is a synchronized and $(Z+1,N)$-bounded distribution. 

  We first show that $\bigoplus_{\rho'\in R}R_{\rho'}, \Psi$ is synchronized. Assume that the effective
  stack height of $c(\rho, i)$ is 1. Let $\rho'' \in
  \bigoplus_{\rho'\in R}R_{\rho'}$, and let $\rho'$ be the element of $R$ it
  corresponds to.

  We have to show that ${\it last}_\Psi(\rho'', i) = {\it
    last}_{\psi_{\rho'}}(\rho'', {\it last}_\psi(\rho', i))$ (which we
  easily deduce from the fact that ${\it last}_\Psi(\rho'', i) = {\it
    last}_\Psi(\rho'', \psi(\rho', {\it last}_\psi(\rho', i)))$. Since $\psi$ is synchronized, 
    we deduce that $c_\Psi (\rho'', i)$ is
  the same configuration as $c_\psi (\rho', i) = c(\rho', {\it
    last}_{\psi}(\rho', i))$ and has effective stack height 1.
  Since $\psi_{\rho'}$ is synchronized, $c(\rho', {\it
    last}_{\psi}(\rho', i))$ is the same configuration as
  $c_{\psi_{\rho'}} (\rho'', {\it last}_{\psi}(\rho', i)) = c(\rho'',
  {\it last}_{\psi_{\rho'}}(\rho'', {\it last}_\psi(\rho', i)))$.
  
   We now prove  that $\bigoplus_{\rho'\in R}R_{\rho'}, \Psi$ is  $(Z+1,N)$-bounded. 
   Again, we pick $\rho''\in
  \bigoplus_{\rho'\in R}R_{\rho'}$, and let $\rho'$ be the corresponding
  element of $R$.  We have to show that $c_\Psi(\rho'', i)$ has effective stack height at most $N$ for  every $0 \leq i \leq Z+ 1$. Since ${\it last}_{\psi}(\rho', Z+1) \leq {\it
    last}_\psi(\rho', Z) + 1$, by monotonicity we deduce that
  ${\it last}_\Psi(\rho'', Z+1) \leq {\it last}_{\psi_{\rho'}} (\rho'',
  {\it last}_\psi(\rho', Z) + 1)$ and we are done.\qed
\end{proof}

\begin{proof}[of Lemma~\ref{lem:Z-N-bounded}, the Boundedness Lemma]
  The proof is by induction on $Z$. If $Z=1$ %
  then $R_0 = \{\rho\}$, because the first configuration of a run has
  effective height at most 2 (if the first rule was a push, and that
  symbol will be later popped). Since by definition $N\geq 2$, we get that
  $\rho$ is $(1,N)$-bounded.

For the induction step, assume that some distribution $R_{Z-1}$ of
$\rho$ is $(Z-1,N)$-bounded and synchronized, and let $\psi$ be the
embedding function of the distribution. If $R_{Z-1}$ is also
$(Z,N)$-bounded, we take $R_Z = R_{Z-1}$, and we are done. Otherwise,
there is $\rho'\in R_{Z-1}$ such that $c_\psi(\rho',0), \ldots,
c_\psi(\rho', Z-1)$ are effectively $N$-bounded, but $c_\psi(\rho',Z)$ is not.

Informally, this means that the $Z$-th transition of $\rho$ was
distributed to $\rho'$. Let $Z_{\rho'}$ be that position in $\rho'$;
formally $Z = \psi (\rho', Z_{\rho'})$ (if no such $Z_{\rho'}$ exists,
$c_\psi(\rho',Z-1) = c_\psi(\rho, Z)$).
Since $Z_{\rho'}$ is the first position of $\rho'$ whose configuration
is not $N$-bounded, we have that $c(\rho',0), \ldots, c(\rho',
Z_{\rho'}-1)$ are $N$-bounded, but $c(\rho',Z_{\rho'})$ is not.
We apply Lemma~\ref{lem:flattening} to each such $\rho'$ and $Z_{\rho'}$, and get 
$(Z_{\rho'},N)$-bounded and synchronized distributions for those $\rho'$: $R_{\rho'} = \{\rho'_a,
\rho'_b\}$.  Let $R_Z$ be the distribution obtained by replacing in
$R_{Z-1}$ every bad run $\rho'$ by $R_{\rho'}$. Then $R_Z$ is an $(Z,N)$-bounded and synchronized distribution
of $\rho$. \qed
\end{proof}

\begin{example}
We give an example showing that that the bound on the effective stack
height used in the Boundedness Lemma is optimal: for any smaller bound, the lemma is no
longer true.  

We build a PDM with $k_1 + k_2 + 1$ states and with stack alphabet $\{\bot\}
\cup [1,k_3]$, where $k_1, k_2, k_3$ are distinct prime numbers. With the $k_1$ first
states, we build a circuit that pushes the word $(1\ldots k_3)^{k_1}$ onto  the
stack. After that, the PDM leaves this circuit, and enters
another one, consisting of $k_2$ states, that pops
$k_2$ stack symbols. The PDM can only leave this circuit from its first state, and only when $\bot$ is the topmost stack symbol; if and when this condition holds, the PDM moves to the last state, from where
it writes victory in the store. It should be clear that, in order to reach the last state, 
the stack of the PDM must reach a height of at least $(1 + k_1 k_2 k_3)$ symbols.
Therefore, no run reaching the last state can be distributed into runs exhibiting a lower effective stack height. 

We now show that we can further improve this example so as to show
that a single instance of the contributor run in parallel with a
special leader may reach the last state, but at least two instances of
its $N$-restriction are required, for at least one of them reaching
that state.

It is possible for the leader to be informed whenever a contributor
takes a loop (once in each loop the contributor informs the leader
through the store and pauses until it receives acknowledgment through
the store). Then the contributor asks permission before entering in
the last state. If the leader only grants permission if he was
informed exactly a multiple of $k_4$ times of the entrance of some
contributor in some loop, then if there is only one contributor, he
may reach the victory state by growing a $1 + k_1k_2k_3k_4$-sized
stack, which is too large for its $N$-restriction. Therefore a single
instance of the $N$-restricted contributor does not suffice. At least
two are required for an accepting run. \qed
\end{example}

\paragraph{Proof of the Reduction Theorem.} Finally, we give the proof of the Reduction Theorem.

\begin{proof}[of Theorem~\ref{th:redPDM}]  Let $w$ be a word of $L(D)$ and let $M$ a be multiset of words of $C$
  compatible with $w$. Let $s$ be a witness of compatibility, and let $\mathcal{I}$ be the
  corresponding interleaving function (as introduced in the proof of the Distributing lemma). 
  Recall that $s$ is an interleaving of $w$ and $M$, that $I(w, i)$ is the position of $s$ at which we find   the $i$-th letter of $w$, and that $I(v, j)$ is the position of $s$ at which we find the $j$-th letter of $v$,
  for every $v \in M$.
  
  Let  $\sigma$ be an accepting run of $w$, and let $R$ be a multiset of runs accepting each element of  the multiset $M$. The proof follows three steps: 
  \begin{itemize}
  \item[(1)] We find a sequence of positions of $s$ corresponding to actions of the leader, such
  that both the run of the leader and $s$ can be pumped between any
  two such positions. 
  \item[(2)] We take a position far enough in this sequence, say $\mathcal{Z}$, and distribute all the runs of $R$ 
  into a multiset $R_\mathcal{Z}$, such that every run of $R_\mathcal{Z}$ is $N$-bounded up to position $Z$.  We show the existence of two positions, say $X,Y$,
  both smaller than $\mathcal{Z}$, satisfying the following condition. Take the multisets of 
  configurations of the runs of $R_\mathcal{Z}$ at positions $X$ and $Y$, and ``prune'' 
  them by removing their dark suffixes. Let $C_X$ and $C_Y$ be the resulting multisets of pruned configurations. Then $C_X$ and  $C_Y$ have the same support (that is, they contain the same elements, although not necessarily the same number of times).
  \item[(3)]  We show that by adding more runs to $R_\mathcal{Z}$, we can obtain a new distribution for which the multisets
   $C_X$ and $C_Y$ not only have the same support, but are equal. We then show that the runs executed by the  leader and by the contributors
    of this new distribution between positions $X$ and $Y$ can be pumped. This yields a word $w_1 w_2^\omega \in L(D)$ (where $w_2$ is the word executed by the leader between positions $X$ and $Y$) compatible with a multiset of words of the form 
    $\{ v_{11} v_{21}^\omega, \ldots, v_{1n} v_{2n}^\omega \}$ (where $v_{21}, \ldots v_{2n}$ are the runs executed by the contributors between positions $X$ and $Y$), and for which we can find a witness of compatibility of the form $s_1 s_2^\omega$, where $s_1$ is an interleaving of $w_1$ and $\{v_{11}, \ldots, v_{1n}\}$, and $s_2$ is an interleaving of $w_2$ and $\{v_{21}, \ldots, v_{2n}\}$
    \end{itemize}

  \noindent {\it Step (1).} Since $\sigma$ is an infinite run of $D$, by Proposition 
  \ref{prop:esh1often} it contains infinitely many positions of effective stack height $1$. 
   By the pigeonhole principle,  from this sequence 
  of positions we can extract an infinite subsequence of configurations with the same
  control state and topmost stack symbol. Since $\sigma$ is also an accepting run of the B\"uchi automaton  $A$, we can further extract  from this sequence an infinite
  subsequence such that between any two positions  an
  accepting state of $A$ is visited. Let $(b_i)$ denote the image of this last infinite sequence
  by $\mathcal{I}$. That is, $(b_i)$ denotes the infinite sequence of positions of $s$ obtained by the procedure above.

  Now from $(b_i)$ we extract a subsequence $(c_i)$ such that between any
  two elements of it, every run of $R$ reaches a configuration with
  effective stack height 1. More formally, for every $i$ and for every $\rho
  \in R$, there exists $p_{i,\rho} \in \dom(\rho)$ such that
  $c(\rho, p_{\rho,i})$ has effective stack height 1 and $c_i <
  \mathcal{I}(\rho, p_{\rho,i}) < c_{i+1}$.  Since, by Proposition 
  \ref{prop:esh1often}, every run of $R$
  reaches infinitely often such configurations, $(c_i)$ exists.  
  This gives us our sequence of positions in $s$.
  
   \medskip
   
  \noindent {\it Step (2)}. Let $t = |M| 2^{|Q_C||\Gamma_C|^{|Q_C|^2|\Gamma_C| + 1}}+1$, and let  $\mathcal{Z} =
  c_t$ (that is, $\mathcal{Z}$ is the position of the $t$-th element of the sequence $(c_i)$). For each run $\rho \in R$, let $Z_\rho$ denote an element of
  $\dom(\rho)$ such that $\mathcal{I}(\rho, Z_\rho) \geq \mathcal{Z}$. By  Lemma~\ref{lem:Z-N-bounded}, we can distribute each run
  $\rho \in R$ into a $(Z_\rho,N)-$bounded multiset $R_\rho$ (with embedding function $\psi_\rho$).

  For every $i \geq 1$, let $q_{\rho,i}$ be the largest position of $\dom(\rho)$ such that $\mathcal{I}(\rho, q_{\rho,i})
  \leq c_i$, and let $R_\rho(q_{\rho,i}) = \{ c_{\psi_\rho}(\tau, q_{\rho,i}) \mid
  \tau \in R_\rho\}$ be  the multiset of configurations of $R_\rho$ at the
  position corresponding to $q_{\rho,i}$. We denote by
  $\alpha_\rho(i)$ the result of removing the dark suffixes of the configurations of 
  $R_\rho(q_{\rho,i}$). We call the result pruned configurations. 
  
  If $i \leq t$, then, by the definition of
  $R_\rho$, all the active prefixes of $R_\rho(q_{\rho,i})$ are
  $N$-bounded. So the pruned configurations of $\alpha_\rho(i)$ consist
  of a control state and a stack content of length at most $N$, and therefore the number of possible pruned 
  configurations is bounded by $|Q_C||\Gamma_C|^{|Q_C|^2|\Gamma_C| + 1}$. It follows that 
  that the number of possible sets (not multisets!) of pruned configurations is strictly smaller than $t$. 
  So by the pigeonhole principle we find two elements $c_l$ and $c_r$
  of the sequence $(c_i)$, where $l < r \leq t$, such that 
  $\alpha_\rho(l)$ and $\alpha_\rho(r)$ have
  the same support for every $\rho$, (i.e. the sets are equal though the multisets may not
  be). 

\medskip

\noindent {\it Step (3).}  We show how to modify the distributions $R_\rho$ so that
  the multisets $\alpha_\rho(l)$ and $\alpha_\rho(r)$ not only have same support, but are 
  equal.  Observe that, even though the multisets are not equal, they have
  the same cardinality. We introduce new runs in the distribution to ``balance'' these multisets. 
  Denote by $\alpha$ the common support of $\alpha_\rho(l)$ and $\alpha_\rho(r)$.
  For every $a\in \alpha$, we find two runs $\rho^l_a$ and $\rho^r_a$ in
  $R_\rho$ such that $c_{\psi_\rho}(\rho^l_a, q_{\rho,l})$ and $c_{\psi_\rho}(\rho^r_a,
  q_{\rho,r})$ have pruned configuration $a$.

  Now we define a new distribution of $\rho$ to a multiset $R_\rho
  \oplus \{\rho_{a,a'} \mid a,a'\in \alpha \}$ with embedding function $\psi'_\rho$. 
  The run $\rho_{a,a'}$ is
  such that the pruned configuration $c_{\psi'_\rho}(\rho_{a,a'}, q_{\rho,l})$ is $a$ and
  $c_{\psi'_\rho}(\rho_{a,a'}, q_{\rho,r})$ is $a'$ : informally
  $\rho_{a,a'}$ does as $\rho^l_a$ up to position $\psi_\rho(\rho^l_a,
  p_{\rho,l})$, and then as $\rho^r_{a'}$ from $\psi_\rho(\rho^r_{a'}, p_{\rho,l})$.
  Formally, $\psi'_\rho$ is the same as $\psi_\rho$ over each
  $\tau \in R_\rho$, and $\psi'_\rho(\rho_{a,a'}, i) =
  \psi_\rho(\rho^l_a, i)$ when $i \leq \psi_\rho(\rho^l_a, p_{\rho,l})$ and
  $\psi'_\rho(\rho_{a,a'}, i) = \psi_\rho(\rho^r_a, i -
  \psi_\rho(\rho^l_a, p_{\rho,l}) + \psi_\rho(\rho^r_{a'}, p_{\rho,l}) )$ when $i >
  \psi_\rho(\rho^l_a, p_{\rho,l})$. Observe that since $c(\rho, p_{\rho,l})$ has
  effective stack height $1$, it is exactly the same configuration as
  $c_{\psi_\rho}(\rho^l_a,p_{\rho,l})$ and $c_{\psi_\rho}(\rho^r_{a'},p_{\rho,l})$.
  It is also the same configuration as $c_{\psi'_\rho}(\rho_{a,a'}, p_{\rho,l})$. So
  $\rho_{a,a'}$ is a run of $C$, and $\psi'_\rho$ is a synchronized
  $(Z_\rho,N)$-bounded distribution of $\rho$.
  By adding to $R_\rho$ sufficiently many instances of the appropriate
  $\rho_{a,a'}$, we obtain a new distribution $R'_\rho$ of $\rho$,
  such that the two multisets $\alpha_\rho(l)$ and
  $\alpha_\rho(r)$ are the same.

  By the Distribution Lemma, the word $w$ is compatible with the words of the runs 
  $\bigoplus_{\rho \in R}  R'_\rho$. 
  Let $\pi$ be a witness of compatibility. Consider the fragment of $\pi$ between the positions corresponding to $c_l$ and $c_r$ in
  $\pi$. The content of the store is the same at these two
  positions. Also, recall that we chose
  the $(c_i)$ so that the projection of the fragment onto the actions
  of the leader can be repeated infinitely often. Denote by
  $w^\circlearrowleft$ the run of the leader consisting of
  repeating the subrun between positions corresponding to $c_l$ and
  $c_r$.  Finally, for each $R'_\rho$, the multiset of pruned configurations
  of $R'_\rho$ at positions $c_l$ and $c_r$ is
  the same, each run in $R'_\rho$ has effective stack height
  $1$ at $c_l$ and $c_r$, and is $N$-bounded on that
  fragment. This does not mean that for every run $\tau \in R_\rho$
  the pruned configuration will be the same at those
  positions, but that there exists a permutation $\mu$ of $R_\rho$ such
  that the pruned configuration of $\tau$ at position $c_l$ 
  is the same as $\mu(\tau)$ at position $c_r$.  Denoting $\ell_\tau =
  (\tau)_{c_{\psi'_\rho}(\tau, p_{\rho,l})+1..c_{\psi'_\rho}(\tau,
    p_{\rho,r})}$, we get that  the multiset
  $\{(\tau)_{1..c_{\psi'_\rho}(\tau, p_{\rho,l})}
  \ell_\tau\ell_{\mu(\tau)}\ell_{\mu^2(\tau)}\ldots \mid \tau \in
  R'_\rho, \rho \in R\}$ is a multiset of $N$-bounded runs, that is
  compatible with $w^\circlearrowleft$. This concludes the proof.
\qed
\end{proof}

%
%
%
%
 %

\end{document}